\documentclass[a4paper,11pt]{article}

\usepackage{geometry}
\usepackage{comment}

\usepackage{amsthm}
\usepackage{amsmath}

\usepackage{amsthm}
\usepackage{comment}
\usepackage{xspace}
\usepackage{amssymb}
\usepackage{graphicx}

\usepackage[shortlabels,inline]{enumitem}
\setlist{noitemsep}
\setlist[enumerate,1]{label=\arabic*., ref=(\arabic*)}

\usepackage{url}
\usepackage{amsmath}
\usepackage{multirow}

\usepackage[tight]{subfigure} 

\usepackage{float}

\usepackage[pdftex]{hyperref}

\graphicspath{{./Images/}}

\usepackage{xcolor}
\definecolor{dark-red}{rgb}{0.4,0.15,0.15}
\definecolor{dark-blue}{rgb}{0.15,0.15,0.4}
\definecolor{medium-blue}{rgb}{0,0,0.5}
\definecolor{gray}{rgb}{0.5,0.5,0.5}

\hypersetup{
    colorlinks, linkcolor={dark-red},
    citecolor={dark-blue}, urlcolor={medium-blue}
}

\usepackage{algorithm}
\usepackage[noend]{algorithmic}
\usepackage{algorithmic-fix}
\renewcommand{\algorithmicrequire}{\textbf{Precondition:}}
\renewcommand{\algorithmicensure}{\textbf{Postcondition:}}

\newcommand{\yes}[0]{\textsc{yes}\xspace}
\newcommand{\no}[0]{\textsc{no}\xspace}

\newcommand{\true}[0]{\textsc{true}\xspace}
\newcommand{\false}[0]{\textsc{false}\xspace}

\newcommand{\containment}[0]{NP~$\subseteq$ coNP$/$poly\xspace}

\newcommand{\F}[0]{\ensuremath{\mathcal{F}}\xspace}
\renewcommand{\P}[0]{\ensuremath{\mathcal{P}}\xspace}

\newcommand{\Q}[0]{\ensuremath{\mathcal{Q}}\xspace}

\newcommand{\G}[0]{\ensuremath{\mathcal{G}}\xspace}
\newcommand{\X}[0]{\ensuremath{\mathcal{X}}\xspace}
\renewcommand{\P}[0]{\ensuremath{\mathcal{P}}\xspace}

\let\plainqed\qedsymbol
\newcommand{\claimqed}{$\lrcorner$}

\newenvironment{claimproof}{\begin{proof}\renewcommand{\qedsymbol}{\claimqed}}{\end{proof}\renewcommand{\qedsymbol}{\plainqed}}

\newcommand{\torso}[0]{\mathop{\mathrm{\textsc{torso}}}}

\newcommand{\Oh}[0]{\ensuremath{\mathcal{O}}\xspace}

\newlength{\baseImageHeight}
\newcommand{\hyphen}{\nobreakdash-\hspace{0pt}}

\newcommand{\nExactSetCover}[0]{\textsc{$n$\hyphen Exact Set Cover}\xspace}
\newcommand{\kPath}[0]{\textsc{$k$\hyphen Path}\xspace}
\newcommand{\kMulticoloredPath}[0]{\textsc{Multicolored $k$\hyphen Path}\xspace}
\newcommand{\kCycle}[0]{\textsc{$k$\hyphen Cycle}\xspace}

\newcommand{\kLeafOutTree}[0]{\textsc{$k$\hyphen Leaf Out-Tree}\xspace}

\newtheorem{observation}{Observation}
\newtheorem{numberedclaim}{Claim}

\newtheorem{theorem}{Theorem}
\newtheorem*{untheorem}{Theorem}
\newtheorem{proposition}{Proposition}
\newtheorem{corollary}{Corollary}
\newtheorem{lemma}{Lemma}

\theoremstyle{definition}

\newtheorem{definition}{Definition}

\date{}

\begin{document}
\title{Turing Kernelization for Finding Long Paths \\ and Cycles in Restricted Graph Classes\thanks{
This work was supported by the Netherlands Organization for Scientific Research (NWO) Veni grant ``Frontiers in Parameterized Preprocessing'' and Gravitation grant ``Networks''. An extended abstract of this work appeared at the 22nd European Symposium on Algorithms (ESA 2014). The present paper contains a streamlined presentation of the full proofs.}}
\author{Bart M.\ P.\ Jansen\\Eindhoven University of Technology, The Netherlands\\\texttt{B.M.P.Jansen@tue.nl}
}

\maketitle

\begin{abstract}
The NP-complete \kPath problem asks whether a given undirected graph has a (simple) path of length at least~$k$. We prove that \kPath has polynomial-size Turing kernels when restricted to planar graphs, graphs of bounded degree, claw-free graphs, or to~$K_{3,t}$-minor-free graphs for some constant~$t$. This means that there is an algorithm that, given a \kPath instance~$(G,k)$ belonging to one of these graph classes, computes its answer in polynomial time when given access to an oracle that solves \kPath instances of size polynomial in~$k$ in a single step. The difficulty of \kPath can therefore be confined to subinstances whose size is independent of the total input size, but is bounded by a polynomial in the parameter~$k$ alone. These results contrast existing superpolynomial lower bounds for the sizes of traditional kernels for the \kPath problem on these graph classes: there is no polynomial-time algorithm that reduces any instance~$(G,k)$ to a single, equivalent instance~$(G',k')$ of size polynomial in~$k$ unless \containment. The same positive and negative results apply to the \kCycle problem, which asks for the existence of a cycle of length at least~$k$. Our kernelization schemes are based on a new methodology called \emph{Decompose-Query-Reduce}.
\end{abstract}

\section{Introduction}
\paragraph{Motivation} Kernelization is a formalization of efficient and provably effective data reduction originating from parameterized complexity theory. In this setting, each instance~$x \in \Sigma^*$ of a decision problem is associated with a parameter~$k \in \mathbb{N}$ that measures some aspect of its complexity. Work on kernelization over the last few years has resulted in deep insights into the possibility of reducing an instance~$(x,k)$ of a parameterized problem to an equivalent instance~$(x',k')$ of size polynomial in~$k$, in polynomial time. By now, many results are known concerning problems that admit such \emph{polynomial kernelization algorithms}, versus problems for which the existence of a polynomial kernel is unlikely because it implies the complexity-theoretic collapse \containment. (See Section~\ref{section:preliminaries} for formal definitions of parameterized complexity.)

In this work we study the possibility of effectively preprocessing instances of the problems of finding long paths or cycles in a graph. In the model of (many-one) kernelization described above, in which the output of the preprocessing algorithm is a single, small instance, we cannot guarantee effective polynomial-time preprocessing for these problems. Indeed, the \kPath and \kCycle problems are \emph{or-compositional}~\cite{BodlaenderDFH09} since the disjoint union of graphs~$G_1, \ldots, G_t$ contains a path (cycle) of length~$k$ if and only if there is at least one input graph with such a structure. Using the framework of Bodlaender et al.~\cite{BodlaenderDFH09} this proves that the problems do not admit kernelizations of polynomial size unless \containment and the polynomial hierarchy collapses to its third level~\cite{Yap83}.

More than five years ago~\cite{BodlaenderDFGHLMRRR08}, the question was raised how fragile this \emph{bad news} is: what happens if we relax the requirement that the preprocessing algorithm outputs a single instance? Does a polynomial-time preprocessing algorithm exist that, given an instance~$(G,k)$ of \kPath, builds a list of instances~$(x_1, k_1), \ldots, \linebreak[1] (x_t, k_t)$, each of size polynomial in~$k$, such that~$G$ has a length-$k$ path if and only if there is at least one \yes-instance on the output list? Such a \emph{cheating kernelization} is possible for the \kLeafOutTree problem~\cite{Binkele-RaibleFFLSV12} while that problem does not admit a polynomial kernelization unless \containment. Hence it is natural to ask whether this can be done for \kPath or \kCycle.

A robust definition of such relaxed forms of preprocessing was given by Lokshtanov~\cite{Lokshtanov09} under the name \emph{Turing kernelization}. It is phrased in terms of algorithms that can query an oracle for the answer to small instances of a specific decision problem in a single computation step.\footnote{Formally, such algorithms are \emph{oracle Turing machines} (cf.~\cite[Appendix A.1]{FlumG06}).} Observe that the existence of an $f(k)$-size kernel for a parameterized problem~$\Q$ shows that~$\Q$ can be solved in polynomial time if we allow the algorithm to make a single size-$f(k)$ query to an oracle for~$\Q$: apply the kernelization to input~$(x,k)$ to obtain an equivalent instance~$(x',k')$ of size~$f(k)$, query the $\Q$-oracle for this instance and output its answer. A natural relaxation, which encompasses the \emph{cheating kernelization} mentioned above, is to allow the polynomial-time algorithm to query the oracle more than once for the answers to $f(k)$-size instances. This motivates the definition of Turing kernelization.

\begin{definition}
Let~$\Q$ be a parameterized problem and let~$f \colon \mathbb{N} \to \mathbb{N}$ be a computable function. A \emph{Turing kernelization for~$\Q$ of size~$f$} is an algorithm that decides whether a given instance~$(x,k) \in \Sigma^* \times \mathbb{N}$ is contained in~$\Q$ in time polynomial in~$|x| + k$, when given access to an oracle that decides membership in~$\Q$ for any instance~$(x',k')$ with~$|x'|, k' \leq f(k)$ in a single step.
\end{definition}

For practical purposes the role of oracle is fulfilled by an external computing cluster that computes the answers to the queries. A Turing kernelization gives the means of efficiently splitting the work on a large input into manageable chunks, which may be solvable in parallel depending on the nature of the Turing kernelization. Moreover, Turing kernelization is a natural relaxation of many-one kernelization that facilitates a theoretical analysis of the nature of preprocessing.

At first glance, it seems significantly easier to develop a Turing kernelization than a many-one kernelization. However, to this date there are only a handful of parameterized problems known for which polynomial-size Turing kernelization is possible but polynomial-size many-one kernelization is unlikely~\cite{AmbalathBHKMPR10,ThomasseTV14,SchaferKMN12,BodlaenderJK14}. Recently, the first \emph{adaptive}\footnote{The algorithm is adaptive because it uses the answers to earlier oracle queries to formulate its next query. In contrast, the cheating kernelization for \kLeafOutTree constructs all its queries without having to know a single answer.} Turing kernelization was given by Thomass\'{e} et al.~\cite{ThomasseTV14} for the \textsc{$k$-Independent Set} problem restricted to bull-free graphs. Although this forms an interesting step forwards in harnessing the power of Turing kernelization, the existence of polynomial-size Turing kernels for \kPath and related subgraph-containment problems remains wide open~\cite{BodlaenderDFGHLMRRR08,Binkele-RaibleFFLSV12,HermelinKSWW15}. Since many graph problems that are intractable in general admit polynomial-size (many-one) kernels when restricted to planar graphs, it is natural to consider whether planarity makes it easier to obtain polynomial Turing kernels for \kPath. This was raised as an open problem by several authors~\cite{Lokshtanov09,MisraRS11}. Observe that, as the disjoint-union argument remains valid even for planar graphs and \kPath is NP-complete in planar graphs, we do not expect polynomial-size many-one kernels for planar \kPath.

\paragraph{Our results} In this paper we introduce the \emph{Decompose-Query-Reduce} framework for obtaining adaptive polynomial-size Turing kernelizations for the \kPath and \kCycle problems on various restricted graph families, including planar graphs and bounded-degree graphs. The three steps of the framework consist of (i) decomposing the input~$(G,k)$ into parts of size~$k^{\Oh(1)}$ with constant-size interfaces between the various parts; (ii) querying the oracle to determine how a solution can intersect such bounded-size parts, and (iii) reducing to an equivalent but smaller instance using this information. In our case, we use a classic result by Tutte~\cite{Tutte66} concerning the decomposition of a graph into its triconnected components, made algorithmic by Hopcroft and Tarjan~\cite{HopcroftT73}, to find a tree decomposition of adhesion two of the input graph~$G$ such that all torsos of the decomposition are triconnected topological minors of~$G$. We complement this with various known graph-theoretic lower bounds on the circumference of triconnected graphs belonging to restricted graph families to deduce that if this Tutte decomposition has a bag of size~$\Omega(k^{\Oh(1)})$, then there must be a cycle (and therefore path) of length at least~$k$ in~$G$. If we have not already found the answer to the problem we may therefore assume that all bags of the decomposition have polynomial size. Consequently we may query the oracle for solutions involving only~$k^{\Oh(1)}$ parts of the decomposition. We use structural insights into the behavior of paths and cycles with respect to bounded-size separators to reduce the number of bags that are relevant to a query to~$k^{\Oh(1)}$. Here we use ideas from earlier work on kernel bounds for structural parameterizations of path problems~\cite{BodlaenderJK13b}. Together, these steps allow us to invoke the oracle to instances of size~$k^{\Oh(1)}$ to obtain the information that is needed to safely discard some pieces of the input, thereby shrinking it. Iterating this procedure, we arrive at a final equivalent instance of size~$k^{\Oh(1)}$, whose answer is queried from the oracle and given as the output of the Turing kernelization. In this way we obtain polynomial Turing kernels for \kPath and the related \kCycle problem (is there a cycle of length \emph{at least}~$k$) in planar graphs, graphs that exclude~$K_{3,t}$ as a minor for some~$t \geq 3$, graphs of maximum degree bounded by~$t \geq 3$, and claw-free graphs. We remark that the \kPath and \kCycle problems remain NP-complete in all these cases~\cite{LiCM00}. Our techniques can be adapted to construct a path or cycle of length at least~$k$, if one exists: for each of the mentioned graph classes~$\G$, there is an algorithm that, given a pair~$(G \in \G,k)$, either outputs a path (respectively cycle) of length at least~$k$ in~$G$, or reports that no such object exists. The algorithm runs in polynomial time when given constant-time access to an oracle that decides \kPath (respectively \kCycle) on~$\G$ for instances of size and parameter bounded by some polynomial in~$k$ that depends on~$\G$.

Our results raise a number of interesting challenges and shed some light on the possibility of polynomial Turing kernelization for the unrestricted \kPath problem. A completeness program for classifying Turing kernelization complexity was recently introduced by Hermelin et al.~\cite{HermelinKSWW15}. They proved that a \emph{colored} variant of the \kPath problem is complete for a class called WK[1] and conjectured that WK[1]-hard problems do not admit polynomial Turing kernels. We give evidence that the classification of the colored variant may be unrelated to the kernelization complexity of the base problem: \kMulticoloredPath remains WK[1]-hard on bounded-degree graphs, while our framework yields a polynomial Turing kernel for (uncolored) \kPath in this case.

\paragraph{Related work} Non-adaptive Turing kernels of polynomial size are known for \kLeafOutTree~\cite{Binkele-RaibleFFLSV12}, \textsc{$k$-Colorful Motif} on comb graphs~\cite{AmbalathBHKMPR10}, and \textsc{$s$-Club}~\cite{SchaferKMN12}. Thomass\'{e} et al.~\cite{ThomasseTV14} gave an adaptive Turing kernel of polynomial size for \textsc{$k$-Independent Set} on bull-free graphs.

\paragraph{Organization} In Section~\ref{section:preliminaries} we give preliminaries on parameterized complexity and graph theory. In Section~\ref{section:cycles} we present Turing kernels for the \kCycle problem. These are technically somewhat less involved than the analogues for \kPath that are described in Section~\ref{section:paths}. While the Turing kernels are phrased in terms of decision problems, we describe how to construct solutions in Section~\ref{section:constructing:solutions}. In Section~\ref{section:multicolored} we briefly consider \kMulticoloredPath.

\section{Preliminaries} \label{section:preliminaries}
\subsection{Parameterized complexity and kernels}
The set~$\{1, 2, \ldots, n\}$ is abbreviated as~$[n]$. For a set~$X$ and non-negative integer~$n$ we use~$\binom{X}{n}$ to denote the collection of size-$n$ subsets of~$X$. A parameterized problem~$\Q$ is a subset of~$\Sigma^* \times \mathbb{N}$, where~$\Sigma$ is a finite alphabet. The second component of a tuple~$(x,k) \in \Sigma^* \times \mathbb{N}$ is called the \emph{parameter}. A parameterized problem is (strongly uniformly) \emph{fixed-parameter tractable} if there exists an algorithm to decide whether $(x,k) \in \Q$ in time~$f(k)|x|^{\Oh(1)}$ where~$f$ is a computable function. 

\begin{definition} \label{def:manyonekernel}
Let~$f \colon \mathbb{N} \to \mathbb{N}$ be a function and~$\Q \subseteq \Sigma^* \times \mathbb{N}$ be a parameterized problem. A \emph{many-one kernelization algorithm} (or \emph{many-one kernel}) for~$\Q$ of size~$f$ is an algorithm that, on input~$(x,k) \in \Sigma^* \times \mathbb{N}$, runs in time polynomial in~$|x| + k$ and outputs an instance~$(x', k')$ such that:
\begin{enumerate}
	\item $|x'|, k' \leq f(k)$, and
	\item $(x,k) \in \Q \Leftrightarrow (x', k') \in \Q$.
\end{enumerate}
It is a \emph{polynomial kernel} if~$f$ is a polynomial (cf.~\cite{Bodlaender09}).
\end{definition}

When used without adjective, the term kernel should be interpreted as a traditional many-one kernel as in Definition~\ref{def:manyonekernel}. We refer to one of the textbooks~\cite{CyganFKLMPPS15,DowneyF13,FlumG06} for more background on parameterized complexity.

\subsection{Graphs}
All graphs we consider are finite, simple, and undirected. An undirected graph~$G$ consists of a vertex set~$V(G)$ and an edge set~$E(G) \subseteq \binom{V(G)}{2}$. We write~$G \subseteq H$ if graph~$G$ is a subgraph of graph~$H$. The subgraph of~$G$ induced by a set~$X \subseteq V(G)$ is denoted~$G[X]$. We use~$G - X$ as a shorthand for~$G[V(G) \setminus X]$. When deleting a single vertex~$v$, we write~$G-v$ rather than~$G - \{v\}$. The \emph{open neighborhood} of a vertex~$v$ in graph~$G$ is~$N_G(v)$. The open neighborhood of a set~$X \subseteq V(G)$ is~$\bigcup _{v \in X} N_G(v) \setminus X$. The \emph{degree} of vertex~$v$ in~$G$ is~$\deg_G(v) := |N_G(v)|$. 

Graph~$H$ is a \emph{minor} of~$G$ if~$H$ can be obtained from a subgraph of~$G$ by contracting edges. Graph~$H$ is a \emph{topological minor} of graph~$G$ if~$H$ can be obtained from a subgraph of~$G$ by repeatedly replacing a degree-2 vertex by a direct edge between its two neighbors. If~$H$ is a minor or a subgraph of~$G$, then~$H$ is also a topological minor of~$G$. Observe that the topological minor relation is transitive.

A vertex of degree at most one is a \emph{leaf}. A \emph{cut vertex} in a connected graph~$G$ is a vertex~$v$ such that~$G - v$ is disconnected. A pair of distinct vertices~$u,v$ is a \emph{separation pair} in a connected graph~$G$ if~$G - \{u,v\}$ is disconnected. A vertex (pair of vertices) is a cut vertex (separation pair) in a disconnected graph if it forms such a structure for a connected component. A graph~$G$ is \emph{biconnected} if it is connected and contains no cut vertices. The \emph{biconnected components} of~$G$ partition the edges of~$G$ into biconnected subgraphs of~$G$. A graph~$G$ is triconnected if removing less than three vertices from~$G$ cannot result in a disconnected graph.\footnote{Some authors require a triconnected graph to contain more than three vertices; the present definition allows us to omit some case distinctions.} A \emph{separation} of a graph~$G$ is a pair~$(A, B)$ of subsets of~$V(G)$ such that~$A \cup B = V(G)$ and~$G$ has no edges between~$A \setminus B$ and~$B \setminus A$. The latter implies that~$A \cap B$ separates the vertices~$A \setminus B$ from the vertices~$B \setminus A$. The \emph{order} of the separation is~$|A \cap B|$. A \emph{minimal separator} in a connected graph~$G$ is a vertex set~$S \subseteq V(G)$ such that~$G - S$ is disconnected and~$G - S'$ is connected for all~$S' \subsetneq S$. A vertex set of a disconnected graph is a minimal separator if it is a minimal separator for one of the connected components.

A \emph{walk} in~$G$ is a sequence of vertices~$v_1, \ldots, v_k$ such that~$\{v_i, v_{i+1}\} \in E(G)$ for~$i \in [k - 1]$. An $xy$-walk is a walk with~$v_1 = x$ and~$v_k = y$. A \emph{path} is a walk in which all vertices are distinct. Similarly, an $xy$-path is an $xy$-walk consisting of distinct vertices. The vertices $x$ and~$y$ are the \emph{endpoints} of an $xy$-path. An $x$-path is a path that has vertex~$x$ as an endpoint. The \emph{length} of a path~$v_1, \ldots, v_k$ is the number of edges on it:~$k-1$. The vertices~$v_2, \ldots, v_{k-1}$ are the \emph{interior vertices} of the path. A \emph{cycle} is a sequence of vertices~$v_1, \ldots, v_k$ that forms a $v_1v_k$-path such that, additionally, the edge~$\{v_1,v_k\}$ is contained in~$G$. The \emph{length} of a cycle is the number of edges on it:~$k$. For an integer~$k$, a $k$-cycle in a graph is a cycle with at least~$k$ edges; similarly a $k$-path is a path with at least~$k$ edges. Throughout the paper we reserve the identifier~$k$ for integers, while we use letters at the end of the alphabet for vertices. This ensures there will be no confusion between $k$-path (a path with at least~$k$ edges), and $x$-path (a path ending in vertex~$x$). The \emph{claw} is the complete bipartite graph~$K_{1,3}$ with partite sets of size one and three. A graph is \emph{claw-free} if it does not contain the claw as an induced subgraph.

\begin{observation} \label{observation:cycleminor:cyclesubgraph}
If a graph contains a cycle (path) of length at least~$k$ as a topological minor, then it contains a cycle (path) of length at least~$k$ as a subgraph.
\end{observation}

\begin{definition}[extension]
Let~$G$ be a graph containing distinct vertices~$x$ and~$y$. Graph~$H$ is an $xy$-extension of~$G$ if~$H$ can be obtained from an induced subgraph of~$G$ by adding an $xy$-path that consists of a single edge, or of new interior vertices of degree two.
\end{definition}

The notion of $xy$-extension will be useful when reasoning about the structure of the graphs for which the oracle is queried by the Turing kernelization.

\begin{proposition} \label{proposition:connect:separator:is:planar}
If~$\{x,y\}$ is a minimal separator in a planar graph~$G$, then any $xy$-extension of~$G$ is planar.
\end{proposition}
\begin{proof}
Consider a planar graph~$G$ with a minimal separator~$\{x,y\}$, and fix an arbitrary plane embedding of~$G$. Since~$\{x,y\}$ is a minimal separator, there are at least two connected components in~$G - \{x,y\}$ and all such components are adjacent to both~$x$ and~$y$ by minimality. We claim that when removing the edges incident on~$u$ and~$v$ from the drawing, the points representing~$u$ and~$v$ belong to the same face. Assume for a contradiction that this is not the case. Then there is a closed curve~$\mathcal{C}$ in the plane (corresponding to a series of edges) not passing through~$x$ or~$y$, that contains~$x$ in its interior while~$y$ is on the exterior. But all edges on such a curve belong to one connected component~$C_1$ of~$G - \{x,y\}$. There is at least one other connected component of~$G - \{x,y\}$, say~$C_2$. Now consider an $xy$-path~$\P$ whose internal vertices all belong to~$C_2$, which exists since~$C_2$ is connected and is adjacent to both~$x$ and~$y$. Since the drawing of~$\P$ connects~$x$ and~$y$, it must cross the closed curve~$\mathcal{C}$. However, as the edges forming~$\mathcal{C}$ belong to~$C_1$, while~$\P$ is a path in~$C_2 \cup \{x,y\}$, their drawings cannot intersect in a valid planar drawing; a contradiction.

Hence vertices~$x$ and~$y$ are indeed in the same face after removing their incident edges from the drawing, which implies that in the original drawing of~$G$ there is a face containing both~$x$ and~$y$. We can draw the edge~$\{x,y\}$, or any~$xy$-path consisting of new degree-2 vertices, in this face without creating crossings. Hence all graphs obtained from~$G$ by adding such a structure are planar. This trivially implies that all graphs obtained from an induced subgraph of~$G$ by adding such a structure are also planar.
\end{proof}

\subsection{Tree decompositions}

The decomposition that is exploited by the Turing kernelization can be described elegantly using tree decompositions. We therefore need the following terminology and simple facts.

\begin{definition} \label{def:treedec}
A \emph{tree decomposition} of a graph~$G$ is a pair~$(T, \X)$, where~$T$ is a tree and~$\X \colon V(T) \to 2^{V(G)}$ assigns to every node of~$T$ a subset of~$V(G)$ called a \emph{bag}, such that:
\begin{enumerate}[(a)]
	\item $\bigcup_{i \in V(T)} \X(i) = V(G)$.\label{td:coverv}
	\item For each edge~$\{u,v\} \in E(G)$ there is a node~$i \in V(T)$ with~$\{u,v\} \subseteq \X(i)$.\label{td:covere}
	\item For each~$v \in V(G)$ the nodes~$\{i \mid v \in \X(i)\}$ induce a connected subtree of~$T$.\label{td:connected}
\end{enumerate}
\end{definition}

The \emph{width} of the tree decomposition is~$\max _{i\in V(T)} |\X(i)| - 1$. The \emph{adhesion} of a tree decomposition is~$\max _{\{i,j\} \in E(T)} |\X(i) \cap \X(j)|$. If~$T$ has no edges, we define the adhesion to be zero. For an edge~$e = \{i,j\} \in E(T)$ we will sometimes refer to the set~$\X(i) \cap \X(j)$ as the \emph{adhesion of edge~$e$}. If~$(T, \X)$ is a tree decomposition of a graph~$G$, then the \emph{torso} of a bag~$\X(i)$ for~$i \in V(T)$ is the graph~$\torso(G, \X(i))$ obtained from~$G[\X(i)]$ by adding an edge between each pair of vertices in~$\X(i)$ that are connected by a path in~$G$ whose internal vertices do not belong to~$\X(i)$.

\begin{observation} \label{observation:connectedset:subtree}
Let~$(T,\X)$ be a tree decomposition of a graph~$G$ and let~$S \subseteq V(G)$ be such that~$G[S]$ is connected. Then the nodes~$\{i \in V(T) \mid \X(i) \cap S \neq \emptyset\}$ induce a connected subtree of~$T$.
\end{observation}

We need the following standard propositions on tree decompositions. We give their proofs for completeness.

\begin{proposition} \label{proposition:neighbors:adhesion}
Let~$(T,\X)$ be a tree decomposition of a graph~$G$ of adhesion at most~$d$, and let~$i \in V(T)$. For each connected component~$C$ of the graph~$G - \X(i)$ we have $|N_G(C) \cap \X(i)| \leq d$.
\end{proposition}
\begin{proof}
Consider a connected component~$C$ of~$G - \X(i)$ and define~$S := \{j \in V(T) \mid \X(j) \cap V(C) \neq \emptyset\}$. By Observation~\ref{observation:connectedset:subtree} the nodes in~$S$ form a connected subtree of~$T$. As~$C$ is a component of~$G - \X(i)$ we have~$i \not \in S$. Assume for a contradiction that~$N_G(C) \cap \X(i)$ contains at least~$d+1$ distinct vertices~$v_1, \ldots, v_{d+1}$. To satisfy condition~\ref{td:covere} of Definition~\ref{def:treedec} for the edges between~$C$ and~$v_1, \ldots, v_{d+1}$, all vertices of~$v_1, \ldots, v_{d+1}$ must occur in a common bag with a vertex of~$C$, hence they must occur in a bag of~$S$. Let~$i'$ be the successor of node~$i$ on the shortest path in~$T$ from node~$i$ to a closest node in~$S$, which is well-defined since~$S$ is a connected subtree. Since~$i \not \in S$ we know that~$i'$ is a neighbor of node~$i$ in~$T$. As~$v_1, \ldots, v_{d+1}$ all occur in the bag of node~$i$, and all occur in a bag of a node in~$S$, Property~\ref{td:connected} of Definition~\ref{def:treedec} implies that~$v_1, \ldots, v_{d+1}$ are all contained in~$\X(i')$. But since these~$d+1$ vertices also occur in~$\X(i)$, this implies that the adhesion of~$(T,\X)$ is at least~$d+1$; a contradiction.
\end{proof}

\begin{proposition}[{\cite[Lemma 7.3]{CyganFKLMPPS15}}] \label{proposition:separation:from:edge}
Let~$(T,\X)$ be a tree decomposition of a graph~$G$, let~$\{i,j\}$ be an edge of the decomposition tree, and let~$T_i$ and~$T_j$ be the trees containing~$i$ and~$j$ respectively, that result from removing the edge~$\{i,j\}$ from~$T$. The pair~$(A,B)$ with~$A := \bigcup _{v \in T_i} \X(v)$ and~$B := \bigcup _{v \in T_j} \X(v)$ is a separation in~$G$ of order~$|\X(i) \cap \X(j)|$.
\end{proposition}

\begin{proposition} \label{proposition:separation:from:children}
Let~$(T,\X)$ be a tree decomposition of a graph~$G$, let~$i$ be a node of the decomposition tree, and let~$j_1, \ldots, j_\ell$ be neighbors of~$i$ such that~$\X(i) \cap \X(j_1) = \X(i) \cap \X(j_2) = \ldots = \X(i) \cap \X(j_\ell) = S$, and let~$T_1, \ldots, T_\ell$ be the trees in the forest~$T - \{i\}$ that contain~$j_1, \ldots, j_\ell$, respectively. Then~$(A,B)$ with~$A := \bigcup _{k=1}^\ell \bigcup _{v \in V(T_k)} \X(v)$ and~$B := (V(G) \setminus A) \cup S$ is a separation in~$G$ of order~$|S|$.
\end{proposition}
\begin{proof}
The preconditions ensure that for all subtrees~$T_j$ with~$j \in [\ell]$, the only vertices of~$G$ that occur in a bag of~$T_j$ and also occur in a bag outside of~$T_j$, are those in~$S$. Since all vertices of~$S$ are in~$\X(j_1)$, this implies that when we remove the edges from~$j_2, \ldots, j_\ell$ to node~$i$, and connect~$j_2,\ldots, j_\ell$ by edges to~$j_1$ instead, the result is a valid tree decomposition~$(T',\X)$ with the same set of bags. Applying Proposition~\ref{proposition:separation:from:edge} to edge~$\{j_1, i\}$ in~$(T',\X)$ yields the proof.
\end{proof}

\begin{proposition} \label{proposition:numbercomponents:numbersubtrees}
Let~$(T,\X)$ be a tree decomposition of a graph~$G$, let~$i \in V(T)$ be a node of the decomposition tree, and let~$U \subseteq V(G)$. If~$G[U] - \X(i)$ has at most~$\ell$ connected components, then there are~$\ell' \leq \ell$ trees~$T_1, \ldots, T_{\ell'}$ in the forest~$T - \{i\}$ such that all nodes~$j$ whose bag~$\X(j)$ contains a vertex of~$U \setminus \X(i)$, are contained in~$\bigcup _{k=1}^{\ell'} V(T_k)$.
\end{proposition}
\begin{proof}
For each connected component~$C$ of~$G[U] - \X(i)$, Observation~\ref{observation:connectedset:subtree} implies that the nodes of~$T$ that contain a vertex of~$C$ form a connected subtree~$T_C$ of~$T$. Since~$C$ is a component of~$G[U] - \X(i)$, node~$i$ is not in~$T_C$. Hence~$T_C$ is contained fully in one of the trees~$T - \{i\}$. Since each of the~$\ell$ components of~$G[U] - \X(i)$ is confined to a single tree of~$T - \{i\}$, the proposition follows.
\end{proof}

\subsection{Tutte decompositions} The following theorem is originally due to Tutte, but has been reformulated in the language of tree decompositions (cf.~\cite[Exercise 12.20]{Diestel10}). For completeness, we give a proof of the current formulation in Appendix~\ref{appendix:tutte:decomposition}. Refer to Figure~\ref{figure:tutte} for an illustration of the involved concepts.

\begin{theorem}[{\cite{Tutte66}}] \label{theorem:tutte}
For every graph~$G$ there is a tree decomposition~$(T,\X)$ of adhesion at most two, called a \emph{Tutte decomposition}, such that:
\begin{enumerate}
	\item for each node~$i \in V(T)$, the graph~$\torso(G, \X(i))$ is a triconnected topological minor of~$G$, and\label{tutte:torsos}
	\item for each edge~$\{i,j\}$ of~$T$ the set~$\X(i) \cap \X(j)$ is a minimal separator in~$G$ or the empty set.\label{tutte:minseparators}
\end{enumerate}
\end{theorem}

An algorithm due to Hopcroft and Tarjan~\cite{HopcroftT73} (see also~\cite{GutwengerM00}) can be used to compute a Tutte decomposition in linear time by depth-first search.\footnote{We remark that the Hopcroft-Tarjan algorithm formally computes triconnected components of a graph, rather than a Tutte decomposition; this corresponds to a variant of Tutte decomposition where each torso is either a triconnected graph or a cycle. A decomposition matching our definition easily follows from their result.} We shall use the following property of Tutte decompositions.

\begin{figure}[t]
\begin{center}
\subfigure[Graph~$G$.]{\label{fig:tutte:graph}
\includegraphics{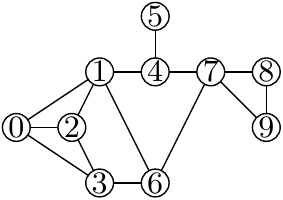}
}
\subfigure[Tutte decomposition~$(T, \X)$ of~$G$.]{\label{fig:tutte:decomposition}
\includegraphics{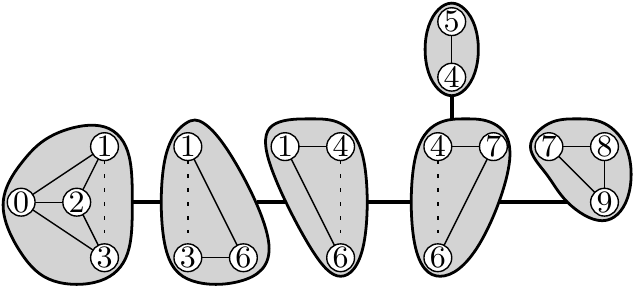}
}
\caption{Example of a Tutte decomposition of a graph. The decomposition tree has six nodes, corresponding to the six gray ovals. Edges of the decomposition tree are visualized as thick lines between ovals. The bag~$\X(i)$ of a node~$i$ is illustrated by drawing the vertices~$\X(i)$ within the oval for node~$i$. The edges within each bag~$i$ represent the torso graph~$\torso(G, \X(i))$. Solid lines represent edges of~$\torso(G, \X(i)) \cap E(G)$, while dotted edges are those that are added by the torso operation.}
\end{center}
\label{figure:tutte}
\end{figure}

\begin{proposition} \label{proposition:minimalsep:in:tutte:edge:in:torso}
Let~$(T,\X)$ be a Tutte decomposition of a graph~$G$. If~$\{x,y\}$ is a minimal separator of~$G$, then for every bag~$\X(i)$ containing~$x$ and~$y$, the edge~$\{x,y\}$ is contained in~$\torso(G, \X(i))$.
\end{proposition}
\begin{proof}
If~$\{x,y\} \in E(G)$ then the proposition is trivial, so assume that this is not the case. Since~$\{x,y\}$ is a minimal separator there are at least two connected components~$C_1,C_2$ of~$G - \{x,y\}$ that are both adjacent to~$x$ and~$y$. Consequently, there is an $xy$-path~$\P_1$ with interior vertices in~$C_1$, and an $xy$-path~$\P_2$ with interior vertices in~$C_2$. We claim that~$\X(i)$ contains vertices from at most one of the components~$C_1$ and~$C_2$. 

Assume for a contradiction that~$v_1 \in \X(i) \cap V(C_1)$ and~$v_2 \in \X(i) \cap V(C_2)$. Since $\torso(G, \X(i))$ is triconnected by the definition of a Tutte decomposition, there is no $v_1v_2$-separator in the torso of size less than three. By Menger's theorem this implies that there are three internally vertex-disjoint $v_1v_2$-paths in the graph~$\torso(G, \X(i))$. Hence there is a $v_1v_2$-path~$\P$ in~$\torso(G, \X(i))$ that contains neither~$x$ nor~$y$. As the adhesion of a Tutte decomposition is at most two, a path in~$\torso(G, \X(i))$ can be expanded into a path in~$G$ by replacing all the shortcut edges that the torso introduces by paths outside~$\X(i)$; this does not change which vertices from~$\X(i)$ are used on the path. Since~$\{x,y\} \subseteq \X(i)$, we can expand~$\P$ into a $v_1v_2$-path in~$G$ that avoids both~$x$ and~$y$. But~$C_1$ and~$C_2$ are distinct connected components of~$G - \{x,y\}$; a contradiction.

Hence~$\X(i)$ contains vertices from at most one of the components~$C_1$ and~$C_2$. Hence at least one of the paths~$\P_1$ or~$\P_2$ is an $xy$-path with interior vertices not in~$\X(i)$, which shows by the definition of torso that~$\torso(G, \X(i))$ contains edge~$\{x,y\}$.
\end{proof}

\subsection{Circumference of restricted classes of triconnected graphs}
The \emph{circumference} of a graph is the length of a longest cycle. Several results are known that give a lower bound on the circumference of a triconnected graph in terms of its order. We will use these lower bounds to deduce that if a Tutte decomposition of a graph has large width, then the graph contains a long cycle (and therefore also a long path).

\begin{theorem} \label{theorem:circumference}
Let~$G$ be a triconnected graph on~$n \geq 3$ vertices and let~$\ell$ be its circumference.
\begin{enumerate}[(a),noitemsep]
	\item If~$G$ is planar, then~$\ell \geq n^{\log_3 2}$.~\cite{ChenY02}
	\item If~$G$ is~$K_{3,t}$-minor free, then~$\ell \geq (1/2)^{t(t-1)} n^{\log _{1729} 2}$.~\cite{ChenYZ12}
	\item If~$G$ is claw-free, then~$\ell \geq (n/12)^{0.753} + 2$.~\cite{BilinskiJMY11}
	\item If~$G$ has maximum degree at most~$\Delta \geq 4$, then~$\ell \geq n^{\log _{r} 2}/2 + 3$, where~$r := \max (64, 4\Delta + 1)$.~\cite{ChenGYZ06}.
\end{enumerate}
\end{theorem}

\subsection{Running times and kernel sizes}
Our algorithms need information about paths and cycles through substructures of the input graph to safely reduce its size without affecting the existence of a solution. Within the framework of Turing kernelization, which is defined with respect to decision oracles that only give \yes/\no answers, we therefore need self-reduction techniques to transform decision algorithms into construction algorithms. The repeated calls to the oracle in the self-reduction contribute significantly to the running time. In practice, it may well be possible to run a direct algorithm to compute the required information (such as the length of a longest $xy$-path, for given~$x$ and~$y$) directly, thereby avoiding the repetition inherent in a self-reduction, to give a better running time. To stay within the formal framework of Turing kernelization we will avoid making assumptions about the existence of such direct algorithms, however, and rely on self-reduction. Since the running time estimates obtained in this way are higher than what would be reasonable in an implementation using a direct algorithm, running time bounds using self-reduction are not very informative beyond the fact that they are polynomial. For this reason we will content ourselves with obtaining polynomial running time bounds in this paper, without analyzing the degree of the polynomial in detail. 

Similar issues exist concerning the size of the kernel, i.e., the size of the instances for which the oracle is queried. For \kCycle on planar graphs, we give explicit size bounds (Theorem~\ref{theorem:planarkcycle}). For \kCycle on other graph families, and for \kPath, we use an NP-completeness transformation to allow the path- or cycle oracle to compute structures such as longest $xy$-paths. These transformations blow up the size of the query instance by a polynomial factor. However, in practice one might be able to use a direct algorithm to compute this information, thereby avoiding the NP-completeness transformation and the associated blowup of kernel size. For this reason it is not very interesting to compute the degree of the polynomial in the kernel size for the cases that NP-completeness transformations are involved. Therefore we only give an explicit size bound for planar \kCycle, where these issues are avoided.

\section{Turing kernelization for finding cycles} \label{section:cycles}
In this section we show how to obtain polynomial Turing kernels for \kCycle on various restricted graph families. After discussing some properties of cycles in Section~\ref{section:cycleproperties}, we start with the planar case in Section~\ref{section:planarkcycle}. In Section~\ref{section:otherkcycle} we show how to adapt the strategy for $K_{3,t}$-minor-free, claw-free, and bounded-degree graphs.

\subsection{Properties of cycles} \label{section:cycleproperties}

We present several properties of cycles that will be used in the Turing kernelization. Recall that a $k$-cycle is a cycle with at least~$k$ edges. The following lemma shows that, after testing one side of an order-two separation for having a $k$-cycle, we may safely remove vertices from that side as long as we preserve a maximum-length path connecting the two vertices in the separator.

\begin{lemma} \label{lemma:cycles:through:separation}
Let~$A, B \subseteq V(G)$ be a separation of order two of a graph~$G$ with~$A \cap B = \{x,y\}$. Let~$V(\P_A)$ be the vertices on a maximum-length $xy$-path~$\P_A$ in~$G[A]$, or~$\emptyset$ if no such path exists. If~$G$ has a $k$-cycle, then~$G[A]$ has a $k$-cycle or~$G[V(\P_A) \cup B]$ has a $k$-cycle.
\end{lemma}
\begin{proof}
Assume that~$G$ has a $k$-cycle~$C$ with edge set~$E(C)$ and vertex set~$V(C)$. If~$V(C) \subseteq A$ then~$G[A]$ contains the $k$-cycle~$C$ and we are done. Similarly, if~$V(C) \subseteq V(\P_A) \cup B$ then the graph~$G[V(\P_A) \cup B]$ contains the $k$-cycle~$C$ and we are done. We may therefore assume that~$C$ contains at least one vertex~$a \in A \setminus B$ and one vertex~$b \in B \setminus A$. Since a cycle provides two internally vertex-disjoint paths between any pair of vertices on it,~$C$ contains two internally vertex-disjoint paths between~$a$ and~$b$. Since~$\{x,y\} = A \cap B$ separates vertices~$a$ and~$b$ by the definition of a separation, each of these two vertex-disjoint paths contains exactly one vertex of~$\{x,y\}$. Hence~$E(C) \cap E(G[A])$ is the concatenation of an~$xa$ and an~$ya$ path in~$G[A]$, and therefore forms an $xy$-path in~$G[A]$. Since~$\P_A$ is a maximum-length $xy$-path in~$G[A]$, the number of edges on~$\P_A$ is at least~$|E(C) \cap E(G[A])|$. Replacing the $xy$-subpath~$E(C) \cap E(G[A])$ of~$C$ by the $xy$-path~$\P_A$ we obtain a new cycle, since all edges of~$G[A]$ that were used on~$C$ are replaced by edges of~$\P_A$. As~$\P_A$ has maximum length, this replacement does not decrease the length of the cycle. Hence the resulting cycle is a $k$-cycle on a vertex subset of~$G[V(\P_A) \cup B]$, which concludes the proof.
\end{proof}

We show how to use an oracle for the decision version of \kCycle to construct longest $xy$-paths, by \emph{self-reduction} (cf.~\cite{FellowsL88}). These paths can be used with the previous lemma to find vertices that can be removed from the graph while preserving a $k$-cycle, if one exists.

\begin{lemma} \label{lemma:selfreduce:xypath}
There is an algorithm that, given an $n$-vertex graph~$G$ with distinct vertices~$x$ and~$y$, and an integer~$k$, either:
\begin{enumerate}
	\item determines that~$G$ contains a~$k$-cycle, or \label{outcome:i:kcycle}
	\item determines that~$G$ contains an $xy$-path of length at least~$k - 1$, or \label{outcome:ii:longpath}
	\item outputs the (unordered) vertex set of a maximum-length $xy$-path in~$G$ (or~$\emptyset$ if no such path exists). \label{outcome:iii:pathvertices}
\end{enumerate}
The algorithm runs in polynomial time when given access to an oracle that decides the \kCycle problem. The oracle is queried for instances~$(G',k)$ with $|V(G')| \leq n + k$, where~$G'$ is an $xy$-extension of~$G$.
\end{lemma}
\begin{proof}
Given an input~$(G,k,x,y)$ we proceed as follows. The algorithm first invokes the \kCycle oracle with the instance~$(G,k)$ to query whether~$G$ has a $k$-cycle. If this is the case, the algorithm reports this and halts with outcome~\ref{outcome:i:kcycle}. If~$x$ and~$y$ belong to different connected components, the algorithm returns the empty set (no $xy$-path exists) and halts with outcome~\ref{outcome:iii:pathvertices}. In the remainder we therefore assume that~$G$ contains an $xy$-path but no $k$-cycle. 

The algorithm adds the edge~$\{x,y\}$ to the graph (if it was not present already) to obtain~$G_0$ and queries whether~$(G_0, k)$ has a $k$-cycle. If this is the case, then~$G$ contains an $xy$-path of length at least~$k-1$: since~$(G_0,k)$ contains a $k$-cycle but~$(G,k)$ does not, the edge~$\{x,y\}$ must be used in any $k$-cycle in~$(G_0, k)$. Removing the edge~$\{x,y\}$ from a $k$-cycle leaves an $xy$-path of length at least~$k-1$. Hence in this case we may report that~$G$ contains an $xy$-path of length at least~$k-1$, according to case~\ref{outcome:ii:longpath}. If~$(G_0,k)$ does not have a $k$-cycle then it is easy to see that the maximum length of an $xy$-path in~$G$ is less than~$k$. The goal of the algorithm now is to identify a maximum-length $xy$-path. The remainder of the procedure consists of two phases: determining the maximum length and finding the path. 

\emph{Determining the length.} To determine the maximum length, we proceed as follows. We create a sequence of graphs~$G_1, \ldots, G_{k-2}$ where~$G_\ell$ is obtained from~$G$ by adding~$\ell$ new vertices~$v_1, \ldots, v_\ell$ to~$G$, along with the edges~$\{v_i, v_{i+1}\}$ for~$i \in [\ell-1]$ and the edges~$\{x,v_1\}$ and~$\{y,v_\ell\}$. The inserted vertices, together with~$x$ and~$y$, form an $xy$-path of length~$\ell+1$. For each graph~$G_\ell$ we invoke the oracle for~$(G_\ell, k)$ to determine whether~$G_\ell$ has a $k$-cycle. Let~$\ell^*$ be the smallest index for which the oracle reports the existence of a $k$-cycle. This is well defined since~$(G_{k-2}, k)$ contains a $k$-cycle that consists of an arbitrary $xy$-path in~$G$ together with the $xy$-path of length~$k-1$ through the new vertices~$v_1, \ldots, v_{k-2}$. Since the circumference of~$G_{i+1}$ is at most the circumference of~$G_i$ plus one, it follows that the circumference of~$G_{\ell^*}$ is \emph{exactly}~$k$ and thus that any $k$-cycle in~$G_{\ell^*}$ has length exactly~$k$; we shall use this fact later. If~$\ell^*$ is the smallest index such that~$G_{\ell^*}$ has a $k$-cycle, then the length of a longest $xy$-path in~$G$ is~$k^* := k - (\ell^* + 1) \geq 1$. Hence by querying the $k$-cycle oracle for the instances~$(G_1, k), \ldots, (G_{k-2},k)$ the algorithm determines the value of~$l^*$ and, simultaneously, the maximum length~$k^*$ of an $xy$-path in~$G$.

\emph{Finding the path.} Using the value of~$l^*$ the algorithm identifies a maximum-length $xy$-path as follows. Set~$H_0 := G_{\ell^*}$. We order the vertices of~$H_0$ from one to~$n + \ell^*$ as~$u_1, \ldots, u_{n + \ell^*}$ and perform the following steps for~$i \in [n + \ell^*]$. Query the oracle for~$(H_{i-1} - u_i, k)$ to determine if~$H_{i-1}$ has a $k$-cycle that does not use~$u_i$. If the oracle answers \yes, let~$H_i := H_{i-1} - u_i$, otherwise let~$H_i := H_{i-1}$. Since~$H_0$ contains a $k$-cycle and the algorithm maintains this property, the final graph~$H_{n + \ell^*}$ has a $k$-cycle. Since vertex~$u_i$ is removed from~$H_{i-1}$ if~$H_{i-1}$ contains a $k$-cycle avoiding~$u_i$, we know that for each vertex in~$H_{n+\ell^*}$ there is no $k$-cycle in~$H_{n+\ell^*}$ without that vertex. As the circumference of~$H_0$ is exactly~$k$, it follows that~$H_{n+\ell^*}$ consists of the vertex set of a cycle of length exactly~$k$; it is a Hamiltonian graph on~$k$ vertices. As~$G_{\ell^*-1}$ does not have a $k$-cycle, all $k$-cycles in~$G_{\ell^*} = H_0$ contain the vertices~$v_1, \ldots, v_{\ell^*}$ on the inserted~$xy$-path, and therefore the graph~$H_{n+\ell^*}$ contains all these vertices. Removing these~$\ell^*$ vertices from~$H_{n+\ell^*}$ yields the vertex set of an $xy$-path in~$G$ of length~$k - (\ell^* + 1) = k^*$, which is a maximum-length $xy$-path in~$G$ as observed above. The vertex set is given as the output for case~\ref{outcome:iii:pathvertices}. 

Let us verify that the oracle queries made by the algorithm are of the required form. The first oracle queries are made for~$G$, and for~$G$ with the edge~$\{x,y\}$ inserted. During the length-determining phase, all query graphs consist of~$G$ with an extra $xy$-path of length at least one (on at most~$k-2$ vertices) inserted. In the second phase, the query graphs consist of induced subgraphs of~$H_0 = G_{\ell^*}$. Since the latter is~$G$ with an extra $xy$-path, the queries indeed take the stated form. Since the total number of queries made by the algorithm is~$\Oh(k)$ in the first phase and~$\Oh(k+n)$ in the second phase, the running time is polynomial using constant-time access to the oracle.
\end{proof}

\begin{figure}[t]
\begin{center}
\subfigure[$A \setminus V(\P) \neq \emptyset$]{\label{fig:largecomponent}
\includegraphics{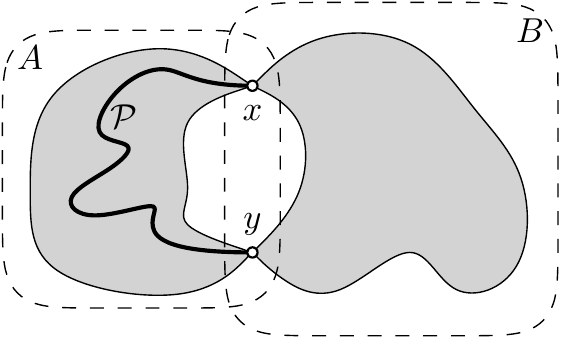}
}
\subfigure[{$G[A] - \{x,y\}$ is disconnected}]{\label{fig:twocomponents}
\includegraphics{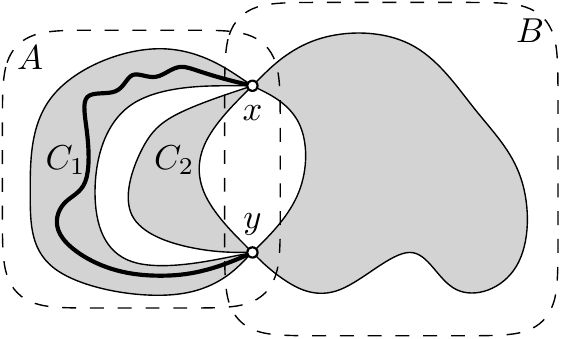}
}
\caption{Schematic illustration of how an instance of $k$-\textsc{Cycle} can be reduced based on a separation~$(A,B)$ of order two, with the corresponding separator~$\{x,y\} = A \cap B$. If~$G[A]$ does not have a $k$-cycle and~$\P$ is a maximum-length $xy$-path in~$G[A]$ (drawn in bold), then the answer to the $k$-\textsc{Cycle} problem is preserved when removing the vertices of~$A \setminus V(\P)$ from the graph. If there is a vertex in~$A \setminus V(\P)$, then this operation shrinks the instance. \ref{fig:largecomponent} If~$|A| \geq k$ and~$|V(\P)| < k$, then the instance is guaranteed to shrink. \ref{fig:twocomponents} If~$G[A] \setminus \{x,y\}$ consists of two connected components~$C_1$ and~$C_2$, then the path~$\P$ is contained entirely within one such component. Removing the vertices of~$A \setminus V(\P)$ therefore eliminates at least one component of~$G[A] - \{x,y\}$ from the input instance.}
\end{center}
\label{figure:reduction:cycle}
\end{figure}

\noindent When the self-reduction algorithm detects a long $xy$-path for a minimal separator~$\{x,y\}$, the following proposition proves that there is in fact a long cycle.

\begin{proposition} \label{proposition:xyseparator:path:gives:cycle}
If~$\{x,y\}$ is a minimal separator of a graph~$G$ and~$G$ contains an $xy$-path of length~$k \geq 2$, then~$G$ contains a $k+1$-cycle.
\end{proposition}
\begin{proof}
Assume the stated conditions hold and let~$\P$ be an $xy$-path of length~$k \geq 2$ in~$G$, which implies it is not a single edge. If~$\{x,y\}$ is an edge of~$G$ then this edge completes~$\P$ into a cycle of length at least~$k+1$ and we are done. Assume therefore that~$\{x,y\}$ is not an edge of~$G$. The interior of the $xy$-path~$\P$, which consists of at least one vertex as~$\P$ has length at least two, is contained entirely within one connected component~$C_\P$ of~$G - \{x,y\}$. Since removal of~$\{x,y\}$ increases the number of connected components (by the definition of minimal separator), there is at least one other connected component~$C'$ of~$G - \{x,y\}$ that is adjacent to vertex~$x$ or vertex~$y$. If~$C'$ is adjacent to only one of~$\{x,y\}$, then that vertex would be a cut vertex, contradicting minimality of the separator~$\{x,y\}$. Component~$C'$ is therefore adjacent to both~$x$ and~$y$ and therefore contains an $xy$-path~$\P'$. Since the interior vertices on this path lie in~$C' \neq C_\P$ it follows that the concatenation of~$\P$ and~$\P'$ is a cycle through~$x$ and~$y$ of length greater than~$k$, which completes the proof.
\end{proof}

\subsection{\texorpdfstring{$k$}{k}-Cycle in planar graphs} \label{section:planarkcycle}

In this section we present the Turing kernelization algorithm for \kCycle on planar graphs. Before giving the overall kernelization routine, we develop the main reduction method. The three statements in Section~\ref{section:cycleproperties} combine into a Turing-style reduction rule for \kCycle instances~$(G,k)$, as given by Algorithm~\ref{alg:reduce-c}. The algorithm works in the general setting where the input graph~$G$ is potentially already partially reduced to an induced subgraph~$G'$; the goal is to reduce~$G'$ further based on a given separation, without changing whether or not it has a $k$-cycle. The following lemma justifies this approach.

\begin{algorithm}[t]
\caption{\textsc{Reduce-C}$(G, G', A, B, x, y, k)$} \label{alg:reduce-c}
\begin{algorithmic}[1]
\REQUIRE{$G'$ is an induced subgraph of~$G$, $(A,B)$ is a separation in~$G'$ with~$A \cap B = \{x,y\}$, and~$\{x,y\}$ is a minimal separator in~$G$.}
\ENSURE{The existence of a $k$-cycle in~$G$ is reported, or the graph~$G'$ is updated by removing all but~$< k$ vertices of~$A \setminus B$. Upon completion, the graph~$G'[A] - \{x,y\}$ has at most one connected component. If~$G'$ initially contained a $k$-cycle, then the deletions preserve this property.}
\vspace{0.2cm}
	\STATE Query the $k$-cycle oracle to determine whether~$G'[A]$ has a $k$-cycle\label{line:kcycle:oracleq:onechild}
	\STATE Find vertices~$S$ of max.~$xy$-path by invoking Lemma~\ref{lemma:selfreduce:xypath} on~$(G'[A], k, x, y)$
	\IF {oracle answers \yes or Lemma~\ref{lemma:selfreduce:xypath} reports $xy$-path of length~$\geq k-1$}
		\STATE Report the existence of a $k$-cycle in~$G$ and halt\label{line:kcycle:answeryes}
	\ELSE
		\STATE Remove the vertices~$A \setminus (S \cup \{x,y\})$ from~$G'$ \label{line:kcycle:restrictbag}
	\ENDIF
\end{algorithmic}
\end{algorithm}

\begin{lemma} \label{lemma:reduce-c:correct}
Algorithm~\ref{alg:reduce-c} satisfies its specifications. It calls the $k$-cycle oracle for $xy$-extensions of~$G'[A]$ with at most~$|A| + k$ vertices.
\end{lemma}
\begin{proof}
Consider the actions of the algorithm on an input satisfying the precondition. We first establish that the algorithm is correct if it reports a $k$-cycle. If the oracle detects a $k$-cycle in~$G'[A]$, then clearly~$G'$ and its supergraph~$G$ have a $k$-cycle as well. If Lemma~\ref{lemma:selfreduce:xypath} yields an $xy$-path of length at least~$k-1$ in~$G'[A]$, then its supergraph~$G$ contains an $xy$-path of length at least~$k-1$. Since~$\{x,y\}$ is a minimal separator in~$G$, Proposition~\ref{proposition:xyseparator:path:gives:cycle} yields a $k$-cycle in~$G$. 

It remains to consider the correctness when no $k$-cycle is detected. Since~$S$ is the vertex set of a maximum-length $xy$-path in~$G'[A]$, the absence of a $k$-cycle in~$G'[A]$ implies by Lemma~\ref{lemma:cycles:through:separation} that~$G'$ has a $k$-cycle if and only if~$G'[S \cup B]$ has a $k$-cycle. Hence the algorithm may safely delete the vertices of~$A \setminus (S \cup \{x,y\})$ without changing the existence of a $k$-cycle. (We take the union with~$\{x,y\}$ to prevent them from being deleted when~$S = \emptyset$, which occurs when there is no $xy$-path.) As the $xy$-path~$S$ has length less than~$k-1$, its vertex set~$|S|$ has less than~$k$ vertices. Hence the algorithm indeed deletes all vertices of~$A \setminus S$ except for less than~$k$ of them. After deleting~$A \setminus (S \cup \{x,y\})$, the only potential connected component of~$G'[A] - \{x,y\}$ is the one containing the interior vertices of the $xy$-path~$S$. Lemma~\ref{lemma:selfreduce:xypath} ensures that the oracle calls are for $xy$-extensions with at most~$|A| + k$ vertices and parameter~$k$. 
\end{proof}

Lemma~\ref{lemma:reduce-c:correct} shows that $k$-cycle instances can be reduced based on suitable separations of order two. Since the size of the oracle queries depends on the $A$-side of the separation, to obtain a polynomial Turing kernel we must reduce the graph based on separations whose $A$-side has size polynomial in~$k$. The key idea behind the following theorem is that either (i) the Tutte decomposition has a large bag, implying by known lower bounds on the circumference of triconnected graphs that the answer to the $k$-cycle problem is \yes, or (ii) we can use the Tutte decomposition to find good separations efficiently.

\begin{theorem} \label{theorem:planarkcycle}
The planar \kCycle problem has a polynomial Turing kernel: it can be solved in polynomial time using an oracle that decides planar \kCycle instances with at most~$(3k + 1) k^{\log _2 3} + k$ vertices and parameter value~$k$.
\end{theorem}
\begin{proof}
We present the Turing kernel for \kCycle on planar graphs following the three steps of the kernelization framework.

\paragraph{Decompose} 
Consider an input~$(G,k)$ of planar \kCycle. First observe that a cycle in~$G$ is contained within a single biconnected component of~$G$. We may therefore compute the biconnected components of~$G$ in linear time using the algorithm by Hopcroft and Tarjan~\cite{HopcroftT73a} and work on each biconnected component separately. In the remainder we therefore assume that the input graph~$G$ is biconnected. By another algorithm of Hopcroft and Tarjan~\cite{HopcroftT73} we can compute a Tutte decomposition~$(T,\X)$ of~$G$ in linear time. For each edge~$\{i,j\} \in E(T)$ of the decomposition tree, the definition of a Tutte decomposition ensures that~$\X(i) \cap \X(j)$ is a minimal separator in~$G$. Since~$T$ has adhesion at most two by Theorem~\ref{theorem:tutte}, these minimal separators have size at most two. Using the biconnectivity of~$G$ it follows that the intersection of the bags of adjacent nodes in~$T$ has size exactly two.

\begin{numberedclaim} \label{claim:planar:decomposition:width}
If there is a node~$i \in V(T)$ of the Tutte decomposition such that~$|\X(i)| \geq k^{\log _2 3}$, then~$G$ has a $k$-cycle.
\end{numberedclaim}
\begin{claimproof}
By the definition of a Tutte decomposition, $\torso(G, \X(i))$ is a triconnected topological minor of~$G$. Since planarity is closed under taking (topological) minors, the torso is planar. Hence the torso is a triconnected planar graph on at least~$k^{\log_2 3}$ vertices, which implies by Theorem~\ref{theorem:circumference} that its circumference is at least~$(k^{\log _2 3})^{\log _3 2} = k$. Consequently, there is a topological minor of~$G$ that contains a $k$-cycle. By Observation~\ref{observation:cycleminor:cyclesubgraph} this implies that~$G$ has a $k$-cycle.
\end{claimproof}

The claim shows that we may safely output \yes if the width of~$(T,\X)$ exceeds~$k^{\log _2 3}$. For the remainder of the kernelization we may therefore assume that~$(T,\X)$ has width at most~$k^{\log _2 3}$. To prepare for the reduction phase we make a copy~$G'$ of~$G$ and a copy~$(T',\X')$ of the decomposition. During the reduction phase we will repeatedly remove vertices from the graph~$G'$ to reduce its size. We will make the convention that vertices that are removed from~$G'$ are implicitly removed from the decomposition~$(T',\X')$, and that we remove leaf nodes of the decomposition tree whose bags are subsets of the bags of their parent. The removals may violate the property of a Tutte decomposition that all torsos of bags are triconnected. However, we will maintain the fact that~$(T',\X')$ is a tree decomposition of adhesion at most two and width at most~$k^{\log _2 3}$ of~$G'$. We root the decomposition tree~$T'$ at an arbitrary vertex to complete the decomposition phase. We use the following terminology. For~$i \in V(T')$ we write~$T'[i]$ for the subtree of~$T'$ rooted at~$i$. For a subtree~$T'' \subseteq T'$ we write~$\X'(T'')$ for the union~$\bigcup _{i \in V(T'')} \X'(i)$ of the bags of the nodes in~$T''$.

\begin{algorithm}[t]
\caption{\textsc{Kernelize-Cycle}$(G, G', (T', \X'), i, k)$} \label{alg:queryreducecycle}
\begin{algorithmic}[1]
\REQUIRE{$G'$ is an induced subgraph of~$G$ with a tree decomposition~$(T',\X')$ of adhesion at most two. A node~$i$ of~$T'$ is specified.}
\ENSURE{The existence of a $k$-cycle in~$G$ is reported, or the graph~$G'$ and decomposition~$(T',\X')$ are updated by removing vertices of~$\X'(T'[i]) \setminus \X'(i)$, resulting in~$|\X'(T'[i])| \leq k \cdot |E(\torso(G, \X(i)))| + |\X(i)|$. If~$G'$ initially contained a $k$-cycle, then the deletions preserve this property.}
\vspace{0.2cm}
	\FOREACH{child~$j$ of~$i$ in~$T'$}
		\STATE Recursively execute \textsc{Kernelize-Cycle}$(G', (T', \X'), j, k)$
		\STATE Let~$\{x,y\} := \X'(i) \cap \X'(j)$
		\STATE \textsc{Reduce-C}$(G, G', A := \X'(T'[j]), B := (V(G') \setminus A) \cup \{x,y\}, x, y, k)$\label{line:kcycle:callonechild}
	\ENDFOR
	\FOREACH{pair $\{x,y\} \in \binom{\X'(i)}{2}$}
		\WHILE{there are distinct children~$j_1, j_2$ of~$i$ in~$T'$ such that~$\X'(i) \cap \X'(j_1) = \X'(i) \cap \X'(j_2) = \{x,y\}$}
			\STATE Let~$A := \X'(T'[j_1]) \cup \X'(T'[j_2])$ and~$B := (V(G') \setminus A) \cup \{x,y\}$
			\STATE \textsc{Reduce-C}$(G, G', A, B, x, y, k)$ \label{line:kcycle:implicitremovechild}
		\ENDWHILE
	\ENDFOR
\end{algorithmic}
\end{algorithm}

\paragraph{Query and reduce} We shrink the instance by repeatedly reducing order-two separations while preserving a $k$-cycle, if one exists. At any point in the process we may detect a $k$-cycle in~$G$ and halt. The main procedure is given as Algorithm~\ref{alg:queryreducecycle}. It is initially called for the root node~$r$ of~$T'$. Intuitively, Algorithm~\ref{alg:queryreducecycle} processes the decomposition tree~$T'$ bottom-up, applying Algorithm~\ref{alg:reduce-c} to two types of separations. During the first \textbf{for each} loop, subtrees~$T'[j]$ rooted at children~$j$ of~$i$ are reduced by attacking separations represented by edge~$\{i,j\}$ of the decomposition tree (see Figure~\ref{fig:largecomponent}). The second \textbf{for each} loop considers the setting where two children have exactly the same adhesion to the current node~$i$ (see Figure~\ref{fig:twocomponents}), and attacks the corresponding separation. If the procedure terminates without reporting a $k$-cycle, we make a final call to the planar \kCycle oracle for the remaining graph~$G'$ and parameter~$k$. The output of the oracle is given as the output of the Turing kernel.

\begin{numberedclaim} \label{claim:kcycle:removefromsubtree}
When the algorithm is called for node~$i$, it only removes vertices belonging to~$\X'(T'[i]) \setminus \X'(i)$.
\end{numberedclaim}
\begin{claimproof}
Vertices are only removed through Algorithm~\ref{alg:reduce-c}. By its postcondition, it only removes vertices of~$A \setminus B$. The~$A$-sides of all relevant separations are subsets of~$\X'(T'[i])$, while the $B$-side always contains~$\X'(i)$. The claim follows.
\end{claimproof}

\begin{numberedclaim} \label{claim:kcycle:abseparation}
When \textsc{Reduce-C} is called, the pair~$(A,B)$ is a separation of~$G'$ and~$A \cap B = \{x,y\}$ is a minimal separator in~$G$.
\end{numberedclaim}
\begin{claimproof}
The fact that~$(A,B)$ is a separation follows from the fact that~$(T',\X')$ is invariantly a tree decomposition of~$G'$, together with Proposition~\ref{proposition:separation:from:edge} (for the first call) and Proposition~\ref{proposition:separation:from:children} (for the second call). It remains to show that~$A \cap B = \{x,y\}$ is a minimal separator. 

By Claim~\ref{claim:kcycle:removefromsubtree}, during the execution of Algorithm~\ref{alg:queryreducecycle} for node~$i$ we only delete vertices of~$G'$ that occur in a bag in the subtree rooted at~$i$, but not in the bag of node~$i$ itself. Hence recursive calls do not remove vertices that belong to bag~$i$, and therefore do not change the intersection between~$i$ and its child bags. Recall that~$(T',\X')$ was initialized as a copy of a Tutte decomposition~$(T,\X)$ of the biconnected graph~$G$, in which all adhesions have size two and are minimal separators of~$G$. It follows that the adhesion between~$i$ and its child bags has size two during the execution for node~$i$ and is equal to the adhesion in the original decomposition~$(T,\X)$. Since every adhesion in a Tutte decomposition of a biconnected graph is a minimal separator by definition, this proves that~$\{x,y\}$ is a minimal separator in~$G$ for all calls to \textsc{Reduce-C}.
\end{claimproof}

By the postcondition of Algorithm~\ref{alg:reduce-c}, each modification step preserves the existence of a $k$-cycle, and the algorithm is correct when it reports a $k$-cycle in~$G$. The oracle answer to the final reduced graph~$G'$ is therefore the correct answer to the original input instance~$(G,k)$. To see that the algorithm runs in polynomial time when given constant-time access to the oracle, the only nontrivial aspect to show is the following claim.

\begin{numberedclaim} \label{claim:kcycle:removechild}
Every time line~\ref{line:kcycle:implicitremovechild} is executed, at least one child subtree of node~$i$ is removed from~$(T',\X')$.
\end{numberedclaim}
\begin{claimproof}
Consider the separation~$(A,B)$ defined within the \textbf{while}-loop based on the children~$j_1$ and~$j_2$ of the current node~$i$. Algorithm~\ref{alg:reduce-c} ensures that~$G'[A] - \{x,y\}$ has at most one connected component~$C$ after the call completes. Let~$U$ be the vertex set of~$C$. By Proposition~\ref{proposition:numbercomponents:numbersubtrees}, at most one tree of~$T' - \{i\}$ has bags containing vertices of~$U$. It follows that at least one of the child subtrees~$T'[j_1]$ and~$T'[j_2]$ contains no vertices of~$U$. All vertices in that subtree except for~$\{x,y\}$ are therefore removed by line~\ref{line:kcycle:implicitremovechild}. Since we implicitly remove leaf nodes of the decomposition whose bag is a subset of their parent bag, the corresponding child subtree disappears from the decomposition~$(T',\X')$.
\end{claimproof}

Claim~\ref{claim:kcycle:removechild} implies that the number of iterations of the \textbf{while}-loop does not exceed the size of the decomposition tree, from which the polynomial-time running time easily follows. The following claim establishes the last part of the postcondition.

\begin{numberedclaim} \label{claim:planar:condition:correct}
When the execution for node~$i$ terminates we have: $$|\X'(T'[i])| \leq \linebreak[1] k \cdot \left |E(\torso(G, \X(i))) \right| + |\X(i)|.$$
\end{numberedclaim}
\begin{claimproof}
By the postcondition, the call to Algorithm~\ref{alg:reduce-c} in the first \textbf{for each} loop removes, for each child~$j$ of~$i$, all but~$< k$ vertices of~$A \setminus B = \X'(T'[j]) \setminus \X'(i)$. Upon completion, each child subtree therefore represents less than~$k$ vertices of~$G'$ that are not in~$\X'(i)$ themselves. The second \textbf{for each} loop repeats while there are at least two children whose bags intersect the bag of~$i$ in the same set of size two. Observe that, since~$G$ was initially biconnected and a recursive call to a child~$j$ does not remove vertices in the intersection of~$j$ to its parent, each bag of a child of~$i$ must have an intersection of size exactly two with the bag of~$i$; this intersection is a minimal separator in~$G$ by Theorem~\ref{theorem:tutte}. Hence upon termination, for each remaining child of~$i$ there is a unique minimal separator~$\{x,y\}$ contained in~$\X'(i) = \X(i)$. By Proposition~\ref{proposition:minimalsep:in:tutte:edge:in:torso}, each such minimal separator yields an edge in~$\torso(G, \X(i))$. Hence the number of children of~$i$ is reduced to~$|E(\torso(G, \X(i)))|$. Since each child represents at most~$k$ vertices of~$G'$ that are not in~$\X'(i)$, while the bag~$\X'(i) = \X(i)$ adds another~$|\X(i)|$ vertices to~$\X'(T'[i])$, it follows that~$|\X'(T'[i])| \leq k \cdot |E(\torso(G, \X(i)))| + |\X(i)|$ upon termination.
\end{claimproof}

Since we are building a Turing kernel for planar \kCycle, the oracle can only decide instances of planar \kCycle. The self-reduction algorithm of Lemma~\ref{lemma:selfreduce:xypath} invoked by the Algorithm~\ref{alg:reduce-c} subroutine only queries instances of the \kCycle problem, but we must still verify that all queried instances are planar. We do this in the next claim, which also establishes the size bound for the queried instances.

\begin{numberedclaim}\label{claim:planar:query:suitable}
Algorithm~\ref{alg:queryreducecycle} only queries the \kCycle oracle with parameter~$k$ on planar graphs of order at most~$(3k + 1) k^{\log _2 3} + k$.
\end{numberedclaim}
\begin{claimproof}
Lemma~\ref{lemma:reduce-c:correct} guarantees that all instances for which the oracle is queried are $xy$-extensions of~$G'[A]$, where~$A$ is the parameter for Algorithm~\ref{alg:reduce-c}. Since~$\{x,y\}$ is the intersection of two adjacent bags in~$(T',\X')$ and therefore also in~$(T,\X)$, by Theorem~\ref{theorem:tutte} the set~$\{x,y\}$ is a minimal separator in~$G$. By Proposition~\ref{proposition:connect:separator:is:planar}, if~$G$ is planar then any $xy$-extension of~$G$ over a minimal separator~$\{x,y\}$ is planar. Since~$G'[A]$ is a subgraph of~$G$, all such extensions of~$G'[A]$ are subgraphs of a planar extension of~$G$, and are therefore planar.

Finally, let us bound the order of the graphs that are queried to the oracle during the execution for some node~$i \in V(T')$. Recall that the width of~$(T',\X')$ is at most~$k^{\log _2 3}$ and therefore that~$|\X'(i)| \leq k^{\log _2 3}$. Since any minor of a planar graph is planar, the graph~$\torso(G', \X'(i))$ is a planar graph on at most~$k^{\log _2 3}$ vertices. Since an $n$-vertex planar graph has at most~$3n$ edges~\cite[Corollary 4.2.10]{Diestel10}, it follows that~$|E(\torso(G', \X'(j))| \leq 3 k^{\log _2 3}$ for all~$j \in V(T')$. Therefore the postcondition of the algorithm guarantees that upon termination for child node~$j$, the number of vertices represented by the subtree rooted at~$j$ is at most~$k \cdot (3k^{\log _2 3}) + k^{\log _2 3} = (3k + 1) k^{\log _2 3}$. This shows that when \textsc{Reduce-C} is invoked in line~\ref{line:kcycle:callonechild}, we have~$|A| \leq (3k + 1) k^{\log _2 3}$. By Lemma~\ref{lemma:reduce-c:correct}, this means it queries the oracle for graphs with at most~$|A| + k \leq (3k + 1) k^{\log _2 3} + k$ vertices. (The same bound applies when the oracle is applied to the final graph~$G'$ after the reduction procedure has finished.) 

When \textsc{Reduce-C} is called in Line~\ref{line:kcycle:implicitremovechild}, each child subtree has already been reduced by the first \textbf{for each} loop and consequently the $A$-side of the separation has at most~$2k$ vertices. Consequently, the resulting oracle queries have at most~$3k$ vertices.
\end{claimproof}

We established that the algorithm outputs the correct answer and satisfies all requirements of a Turing kernelization, concluding the proof of Theorem~\ref{theorem:planarkcycle}.
\end{proof}

\subsection{\texorpdfstring{$k$}{k}-Cycle in other graph families} \label{section:otherkcycle}

There are two obstacles when generalizing the Turing kernel for \kCycle on planar graphs to the other graph families. In the decompose step we have to ensure that each torso of the Tutte decomposition still belongs to the graph family, so that Theorem~\ref{theorem:circumference} may be used to deduce the existence of a $k$-cycle if the width of the Tutte decomposition is sufficiently large. Lemma~\ref{lemma:torsos:restricted} is used for this purpose. In the query step we have to deal with the fact that the alterations made to the graph by the self-reduction procedure may violate the defining property of the graph class, which can be handled by using an NP-completeness transformation before querying the oracle. Besides these issues, the kernelization is the same as in the planar case.

\begin{lemma} \label{lemma:torsos:restricted}
Let~$(T,\X)$ be a Tutte decomposition of a graph~$G$, let~$i \in V(T)$, and let~$H$ be a graph.
\begin{enumerate}
	\item If~$G$ has maximum degree~$\Delta$, then~$\torso(G, \X(i))$ has maximum degree at most~$\Delta$.\label{torsos:maxdegree}
	\item If~$G$ is $H$-minor-free, then~$\torso(G, \X(i))$ is $H$-minor-free.\label{torsos:minorfree}
	\item If~$G$ is claw-free, then~$\torso(G, \X(i))$ is claw-free.\label{torsos:clawfree}
\end{enumerate}
\end{lemma}
\begin{proof}
The key point is that, by definition of the Tutte decomposition, every graph $\torso(G, \X(i))$ is a topological minor of~$G$. As taking a topological minor (deleting edges/vertices and replacing degree-2 vertices by an edge) cannot increase the degree of a vertex, this implies~\ref{torsos:maxdegree}. If~$\torso(G, \X(i))$ contains~$H$ as a minor, then a topological minor of~$G$ contains an $H$-minor, showing that~$G$ contains an $H$-minor. Hence contraposition gives~\ref{torsos:minorfree}. It remains to establish~\ref{torsos:clawfree}.

Assume for a contradiction that~$G$ is claw-free, but~$\torso(G, \X(i))$ has a claw (induced $K_{1,3}$ subgraph) with center~$v \in \X(i)$ and leaves~$u_1,u_2,u_3 \in \X(i)$ for some~$i \in V(T)$. Let~$E^* := \{ \{v,u_1\}, \{v,u_2\}, \{v,u_3\} \} \setminus E(G)$ be the edges used in the claw that are not present in~$G$; these were added by the torso operation. Since~$\torso(G, \X(i))$ contains all edges of~$G[\X(i)]$, we know that~$\{u_1, u_2, u_3\}$ is an independent set in~$G$ and at least one of these vertices is not adjacent to~$v$ in~$G$ (as otherwise~$G$ would have a claw.) Hence~$E^*$ is nonempty. 

Consider an edge~$\{v,u_j\} \in E^*$. Since this edge was added by the torso operation, there is a $vu_j$-path in~$G$ whose internal vertices avoid~$\X(i)$, and therefore belong to some connected component of~$G - \X(i)$. Accordingly, let~$C_j$ be a component containing the interior vertices of a~$vu_j$-path for each~$\{v,u_j\}$ in~$E^*$. We argue that the components~$C_j$ are all distinct. Suppose that some component~$C^*$ contains the interior vertices of both a~$vu'$-path and a $vu''$-path for distinct~$\{v,u'\}, \{v,u''\} \in E^*$. Then the connected component~$C^*$ of~$G - \X(i)$ is adjacent to the three vertices~$v,u',u'' \in \X(i)$. By Proposition~\ref{proposition:neighbors:adhesion} this implies that the adhesion is at least three, contradicting the fact that a Tutte decomposition has adhesion at most two. Hence the components~$C_j$ for~$\{v,u_j\} \in E^*$ are all distinct. For each~$\{v,u_j\} \in E^*$ let~$w_j$ be the successor of~$v$ in a~$vu_j$-path through~$C_j$. Since the components~$C_j$ are all distinct, the chosen vertices~$w_j$ are all distinct. Since the vertices~$w_j$ belong to different connected components of~$G - \X(i)$, they are mutually non-adjacent. By Proposition~\ref{proposition:neighbors:adhesion}, no vertex of~$\X(i) \setminus \{v, u_j\}$ is adjacent to~$w_j$. Hence we may replace each edge~$\{v,u_j\} \in E^*$ in the claw by~$\{v,w_j\}$ to obtain a claw in~$G$; a contradiction to the assumption that~$G$ is claw-free.
\end{proof}

\begin{theorem} \label{theorem:otherkcycle}
The \kCycle problem has a polynomial Turing kernel when restricted to graphs of maximum degree~$t$, claw-free graphs, or $K_{3,t}$-minor-free graphs, for each constant~$t \geq 3$.
\end{theorem}
\begin{proof}
The approach is similar to that of Theorem~\ref{theorem:planarkcycle}. Observe that a biconnected component of a graph~$G$ is an induced subgraph of~$G$. Since all mentioned graph classes are hereditary, it follows that if~$G$ belongs to one of the mentioned classes, then all its biconnected components do as well. Consequently, we may again apply the Turing kernelization algorithm to all biconnected components individually. By Lemma~\ref{lemma:torsos:restricted} it follows that for each mentioned graph family~$\G$, the torso of a node~$i$ of a Tutte decomposition of~$G \in \G$ belongs to the same family~$\G$. By using the subresult of Theorem~\ref{theorem:circumference} corresponding to the particular choice of graph class we therefore establish the required analogue of Claim~\ref{claim:planar:decomposition:width}: if the width of a Tutte decomposition is not bounded by a suitable polynomial in~$k$, then~$G$ has a $k$-cycle and we may answer \yes. We can apply Algorithm~\ref{alg:queryreducecycle} without modifications to recursively reduce the instance. Claims~\ref{claim:kcycle:removefromsubtree}--\ref{claim:planar:condition:correct} continue to hold. We only have to change the argumentation for Claim~\ref{claim:planar:query:suitable}, since the graphs that will be queried to the oracle will not be planar and will be larger than in the planar case.

We show that the size of the queried instances is polynomial in~$k$. By induction, the postcondition of the algorithm ensures that in the execution for node~$i$, the recursive calls for the child nodes~$j$ decrease the number of vertices represented in the subtrees~$T'[j]$ to~$k \cdot |E(\torso(G, \X(j)))| + |\X(j)|$. Since the width of the decomposition is polynomial in~$k$ and the number of edges of a graph is quadratic in its order, the sizes of the child subtrees are reduced to a polynomial in~$k$ that depends on the graph class and its parameters. Since the $A$-sides of the separations defined in Algorithm~\ref{alg:queryreducecycle} consist of one or two child subtrees that have already been reduced recursively, this implies that the Algorithm~\ref{alg:reduce-c} subroutine only queries the oracle for graphs of size polynomial in~$k$. 

It remains to consider the type of instances for which the oracle is queried during the procedure. Lemma~\ref{lemma:reduce-c:correct} ensures that the oracle is only queried for $xy$-extensions of~$G'[A]$, where~$\{x,y\}$ is a minimal separator in~$G$ by Claim~\ref{claim:kcycle:abseparation}. Unfortunately, this is not sufficient to guarantee that the query graphs belong to the same graph class as the input graph. While an $xy$-extension of~$G'[A]$ does not have larger maximum degree than~$G$ and cannot have larger~$K_{3,t}$ minors than~$G$, it may be that an $xy$-extension of~$G'[A]$ has a claw whereas~$G$ was claw-free. In particular, this can happen if the vertices~$\{x,y\}$ are connected by an edge in~$G$. 

Rather than trying to find an ad-hoc workaround for this issue, we adopt the following robust solution. Recall that the classical version of the \kCycle problem is NP-complete for all graph classes mentioned in the theorem~\cite{LiCM00}. The algorithm of Lemma~\ref{lemma:selfreduce:xypath} needs to query the oracle for the answers to instances~$(H,k)$ of the \kCycle problem (on an unrestricted graph). Since \kCycle is in NP, the NP-completeness transformation from general \kCycle to the \kCycle problem restricted to the relevant graph class can be used to transform instance~$(H,k)$ in polynomial time into an equivalent instance~$(H',k')$. Since a polynomial-time transformation cannot blow up the instance size superpolynomially, the order of~$H'$ is polynomial in the order of~$H$, which is polynomial in~$k$ in all applications of Lemma~\ref{lemma:selfreduce:xypath}. We can therefore modify the algorithm as follows: whenever the algorithm tries to query the \kCycle oracle, we first use the NP-completeness transformation to obtain an equivalent $k'$-\textsc{Cycle} instance on the appropriate graph class, convert it to a parameterized instance, and query that instead. By this adaptation the oracle is only queried for instances that it can answer. The size and parameter of the queried instances remain polynomial in~$k$. As this resolves the last issue, this completes the proof of Theorem~\ref{theorem:otherkcycle}.
\end{proof}

\section{Turing kernelization for finding paths} \label{section:paths}
Now we turn our attention to the \kPath problem. While the main ideas are the same as in the \kCycle case, the details are a bit more technical, for two reasons. Since a path may cross several biconnected components, we can no longer restrict ourselves to biconnected graphs and therefore the minimal separators formed by the intersections of adjacent bags of the Tutte decomposition may now have size one or two. Additionally, there are several structurally different ways in which a path may cross a separation of order two and we have to account for all possible options. To query for the relevant information, we need a more robust self-reduction algorithm. We first develop the structural claims and self-reduction tools in Section~\ref{section:properties:paths}. In Section~\ref{section:kpath:turing} we present the Turing kernels.

\subsection{Properties of paths} \label{section:properties:paths}

The following two statements describe how longest paths intersect separations of order one and two. The order-one case is easily summarized by the following observation.

\begin{observation} \label{observation:paths:through:cutvertex}
Let~$A,B \subseteq V(G)$ be a separation of order one of a graph~$G$ with~$A \cap B = \{x\}$. Let~$V(\P_A)$ be the vertices on a maximum-length path~$\P_A$ in~$G[A]$ that ends in~$x$. If~$G$ has a $k$-path, then~$G[A]$ has a $k$-path or~$G[V(\P_A) \cup B]$ has a $k$-path.
\end{observation}

Recall that for a vertex~$x$, an $x$-path is a path that has~$x$ as an endpoint. The six different types of witness structures described in the following lemma are illustrated in Figure~\ref{figure:structures:path}.

\begin{figure}[t]
\begin{center}
\subfigure[Subgraph~$\P_1$.]{
\includegraphics[scale=0.9]{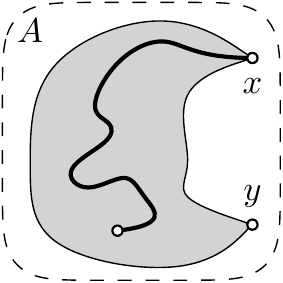}
}
\subfigure[Subgraph~$\P_3$.]{
\includegraphics[scale=0.9]{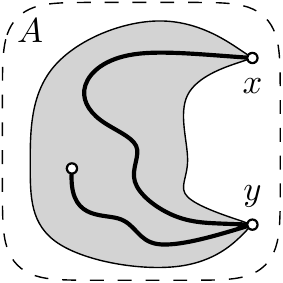}
}
\subfigure[Subgraph~$\P_5$.]{
\includegraphics[scale=0.9]{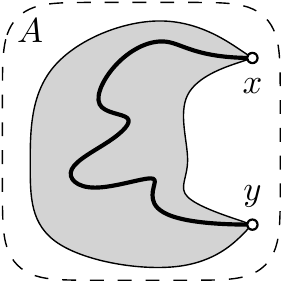}
}
\subfigure[Subgraph~$\P_6$.]{
\includegraphics[scale=0.9]{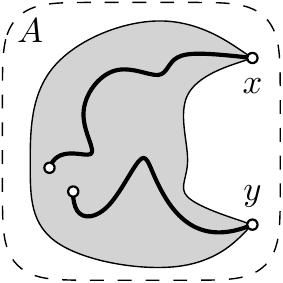}
}
\caption{Schematic illustration of the distinct ways in which a maximum-length path can intersect one side of an order-two separation with separator~$\{x,y\}$. For compactness, only the $A$-side of the separation is shown. The subgraphs~$\P_1,\P_3,\P_5$, and~$\P_6$ are represented by thick curves. They are described in Lemma~\ref{lemma:paths:through:separation}. Subgraphs~$\P_2$ and~$\P_4$ are mirror images of~$\P_1$ and~$\P_3$, respectively.}
\label{figure:structures:path}
\end{center}
\end{figure}

\begin{lemma} \label{lemma:paths:through:separation}
Let~$A, B \subseteq V(G)$  be a separation of order two of a graph~$G$ with~$A \cap B = \{x,y\}$. Let~$\P_1, \ldots, \P_6$ be subgraphs of~$G[A]$ such that:
\begin{enumerate}
	\item $\P_1$ is a maximum-length $x$-path in~$G[A] - \{y\}$.
	\item $\P_2$ is a maximum-length $y$-path in~$G[A] - \{x\}$.
	\item $\P_3$ is a maximum-length $x$-path in~$G[A]$.
	\item $\P_4$ is a maximum-length $y$-path in~$G[A]$.
	\item $\P_5$ is a maximum-length $xy$-path in~$G[A]$, or~$\emptyset$ if no such path exists.
	\item $\P_6$ consists of two vertex-disjoint paths in~$G[A]$, one $x$-path and one $y$-path, such that the combined length of these paths is maximized.
\end{enumerate}
If~$G$ has a $k$-path, then~$G[A]$ has a $k$-path or~$G[(\bigcup _{i = 1}^6 V(\P_i)) \cup B]$ has a $k$-path.
\end{lemma}
\begin{proof}
Consider a $k$-path~$\P$ in~$G$. Let~$\P_A$ be the subgraph of~$\P$ consisting of the edges~$E(\P) \cap E(G[A])$. Similarly, let~$\P_B$ be the subgraph of~$\P$ consisting of the edges~$E(\P) \cap (E(G[B]) \setminus E(G[A]))$ such that every edge on~$\P$ is contained in exactly one of~$\P_A$ and~$\P_B$. Observe that the lemma is trivial if~$\P$ is contained within~$G[A]$ or within~$G[B]$. In the remainder we may therefore assume that~$\P$ contains a vertex~$a \in A \setminus B$ and a vertex~$b \in B \setminus A$. Since~$\{x,y\}$ is the separator corresponding to the separation~$(A,B)$, path~$\P$ must traverse at least one vertex of~$\{x,y\}$ to connect~$a$ and~$b$. This implies that there is a vertex~$z \in \{x,y\}$ such that~$\deg_{\P_A}(z) = \deg_{\P_B}(z) = 1$; in particular, we can choose~$z$ by starting at vertex~$a$ and traversing the path until the first time it is about to visit a vertex in~$B \setminus A$; observe that this even holds if~$\{x,y\}$ is an edge of~$G$ that is contained in~$\P$. Since the set of subgraphs~$\P_1, \ldots, \P_6$ is symmetric with respect to~$x$ and~$y$, we may assume without loss of generality that~$\deg_{\P_A}(x) = \deg_{\P_B}(x) = 1$. We proceed by a case distinction on the values~$\deg_{\P_A}(y)$ and~$\deg_{\P_B}(y)$. Observe that~$\deg_{\P_A}(y) + \deg_{\P_B}(y) \leq 2$ since the subgraphs~$\P_A$ and~$\P_B$ partition~$\P$ and a vertex on a path has at most two incident edges on that path.

\begin{enumerate}
	\item If~$\deg_{\P_A}(y) = \deg_{\P_B}(y) = 1$, we distinguish two cases:
	\begin{enumerate}
			\item If~$\P_A$ is a connected subgraph of~$\P$, then since the vertices~$x$ and~$y$ have degree one in this subgraph (their other incident edges on the path~$\P$ are contained in subgraph~$\P_B$) the subgraph~$\P_A$ forms an $xy$-path in~$G[A]$. Replacing this $xy$-subpath of~$\P$ by the maximum-length $xy$-path~$\P_5$ in~$G[A]$ we therefore obtain a path that is at least as long, proving the existence of a $k$-path in~$G[(\bigcup _{i = 1}^6 V(\P_i)) \cup B]$.
			\item Now assume that~$\P_A$ is a disconnected subgraph of~$\P$. Each connected component of~$\P_A$ contains one of the vertices~$\{x,y\}$, since~$\P$ is connected and these are the only vertices in~$G[A]$ for which some of their incident edges in~$G$ are not contained in~$G[A]$. Since~$\deg_{\P_A}(x) = \deg_{\P_A}(y) = 1$, there are exactly two connected components in~$\P_A$ and each component contains one of~$x,y$ as a degree-one vertex. (Since~$\P_A$ would be connected if~$\{x,y\}$ would be an edge on~$\P$, we know that~$\{x,y\} \not \in \P$.) Hence one of the components of~$\P_A$ is an $x$-path and the other one is an $y$-path. Since the combined length of these two paths is at most the combined length of the $x$-path and $y$-path in~$\P_6$, we can replace~$\P_A$ by~$\P_6$ to obtain a $k$-path in~$G[(\bigcup _{i = 1}^6 V(\P_i)) \cup B]$. 
	\end{enumerate}
	\item If~$\deg_{\P_A}(y) \geq 1$, then by the case above we have~$\deg_{\P_B}(y) = 0$. Since~$\P$ can only cross the separator~$\{x,y\}$ at vertex~$x$ (as~$\deg_{\P_B}(y) = 0$), the restriction of~$\P$ to~$G[A]$ consists of a single connected component which forms an $x$-path in~$G[A]$. Since~$\deg_{\P_B}(y) = 0$, vertex~$y$ is not used on~$\P_B$. Therefore we can replace the $x$-path~$\P_A$ in~$\P$ by the $x$-path~$\P_3$ to obtain a new path; by the maximality of~$\P_3$, this path is at least as long as~$\P$ which proves that~$G[(\bigcup _{i = 1}^6 V(\P_i)) \cup B]$ contains a $k$-path.
	\item Otherwise we have~$\deg_{\P_A}(y) = 0$, which implies that~$\P_A$ is an $x$-path in~$G[A]$ that does not contain the vertex~$y$. Replacing~$\P_A$ by~$\P_1$ we therefore obtain a path that is at least as long and which is contained in~$G[(\bigcup _{i = 1}^6 V(\P_i)) \cup B]$. Hence the latter graph contains a $k$-path.
\end{enumerate}

As the cases are exhaustive, this concludes the proof of Lemma~\ref{lemma:paths:through:separation}.
\end{proof}

Now we turn to the self-reduction that is needed for \kPath. The self-reduction procedure of Lemma~\ref{lemma:selfreduce:xypath} suffices to obtain a Turing kernel for the \kCycle case. The Turing kernelization for \kCycle queries the oracle to compute longest $xy$-paths. In the case of \kPath, Lemma~\ref{lemma:paths:through:separation} shows that we will need other information besides just a maximum $xy$-path. To avoid having to construct ad-hoc self-reductions for the various pieces of information needed in the lemma, we give a general theorem that shows how queries to an oracle for an arbitrary NP-complete language may be used to find maximum-size subgraphs specifying certain properties. We will need the following terminology.

\begin{definition}
A \emph{$2$-terminal graph} is a triple~$(G,x,y)$ where~$G$ is a graph and~$x,y$ are distinguished terminal vertices in~$G$ that are not necessarily distinct. A \emph{stable $2$-terminal edge property} is a function~$\Pi$ which takes as parameters a $2$-terminal graph~$(G,x,y)$ and an edge subset~$Y \subseteq E(G)$ and outputs \true or \false, such that the following holds: if~$\Pi((G,x,y), Y) = \true$ then for any subgraph~$G'$ of~$G$ that contains~$x$,~$y$, and the edge set~$Y$, we have~$\Pi((G',x,y),Y) = \true$. 
\end{definition}

For example, observe that the properties ``the edge set~$Y$ forms a path between~$x$ and~$y$'' and ``the edge set~$Y$ consists of two vertex-disjoint paths, one ending in~$x$ and one ending in~$y$'' are stable $2$-terminal edge properties. The following lemma shows how to find maximum-size edge sets satisfying a stable $2$-terminal edge property by self-reduction. 

For a parameterized problem~$\Q \subseteq \Sigma^* \times \mathbb{N}$, denote by~$\widetilde{Q}$ the classical language~$\widetilde{\Q} := \{ x\#1^k \mid (x,k) \in \Q\}$, where~$\#$ is a new character that is added to the alphabet.

\begin{lemma} \label{lemma:selfreduce:stableproperty}
Let~$\Q$ be a parameterized problem such that~$\widetilde{\Q}$ is NP-complete. Let~$\Pi$ be a polynomial-time decidable stable $2$-terminal edge property. There is an algorithm that, given a $2$-terminal graph~$(G,x,y)$, computes a maximum-cardinality set~$Y \subseteq E(G)$ that satisfies~$\Pi$, or determines that no nonempty edge set satisfies~$\Pi$. The algorithm runs in polynomial time when given access to an oracle that decides instances of~$\Q$ with size and parameter polynomial in~$|V(G)|$ in constant time.
\end{lemma}
\begin{proof}
The overall proof strategy is similar to that of Lemma~\ref{lemma:selfreduce:xypath} in that we first determine the maximum cardinality of a set that has the property and then use self-reduction to find it. The difference is that we have to use the NP-completeness transformation to~$\widetilde{\Q}$ to make our queries to an oracle for~$\Q$, rather than an oracle that decides \kCycle.

\paragraph{Determining the maximum size} Consider following decision problem~$L_\Pi$: given a $2$-terminal graph $(G,x,y)$ and an integer~$k$, is there an edge set~$Y$ of size at least~$k$ such that~$\Pi((G,x,y),Y) = \true$? This decision problem is contained in NP since a non-deterministic algorithm can decide an instance in polynomial time by guessing an edge set~$Y$ of size at least~$k$ and then checking whether it satisfies~$\Pi$; the latter can be done in polynomial time by our assumption on~$\Pi$. Since~$L_\Pi$ is contained in NP and~$\widetilde{Q}$ is NP-complete by assumption, there is a polynomial-time computable function~$f$ such that for any string~$s \in \Sigma^*$ we have~$s \in L_{\Pi}$ if and only if~$f(s) \in \widetilde{Q}$. As the transformation is polynomial-time, it cannot output a string of length superpolynomial in the input size and therefore~$|f(s)|$ is polynomial in~$|s|$. Observe that we can easily split a well-formed instance~$f(s) = z \# 1^k$ of~$\widetilde{Q}$ into an equivalent parameterized instance~$(z,k)$. Since the value of~$k$ is encoded in unary in instances of~$\widetilde{Q}$, the size~$|z|$ of the parameterized instance and its parameter~$k$ are both bounded by a polynomial in~$|s|$. Using these transformations together with the oracle for~$\Q$, we can obtain the answer to any instance~$s$ of~$L_\Pi$ by querying the $\Q$-oracle for a parameterized instance of size and parameter bounded polynomially in~$|s|$.

These observations allow us to determine the maximum cardinality of an edge set satisfying~$\Pi$, as follows. Let~$m$ be the number of edges in~$G$. For~$i \in [m]$ we consider the instance~$s_i = ((G,x,y), i)$ of problem~$L_\Pi$, which asks whether~$(G,x,y)$ has an edge set of size at least~$i$ satisfying~$\Pi$. We transform each instance~$s_i$ of~$L_\Pi$ to an instance of~$\Q$ and query the $\Q$-oracle for the corresponding instance. Observe that the queried instances have size polynomial in~$n := |V(G)|$. If the oracle only answers \no then there is no nonempty edge set satisfying~$\Pi$ and we output this. Otherwise we let~$k^*$ be the largest index for which~$f(s_i)$ is a \yes-instance of~$\widetilde{Q}$, which is clearly the maximum cardinality of a subset satisfying~$\Pi$.

\paragraph{Finding a maximum-cardinality set} Using the value of~$k^*$ we use self-reduction to find a corresponding satisfying edge set of size~$k^*$. Number the edges in~$G$ as~$e_1, \ldots, e_m$. Let~$H_0 := G$. For~$i \in [m]$, perform the following steps. Query the $\Q$-oracle for the parameterized instance corresponding to~$f((H_{i-1} - \{e_i\}, x,y), k^*)$. The oracle then determines whether the graph obtained from~$H_{i-1}$ by removing edge~$e_i$ has a set of size~$k^*$ satisfying~$\Pi$. If the oracle answers \yes, then define~$H_i := H_{i-1} - \{e_i\}$; otherwise let~$H_i := H_{i-1}$. Since~$H_0 = G$ has a satisfying set of size~$k^*$, the procedure maintains the invariant that~$H_i$ contains a satisfying edge set of size~$k^*$. By the definition of a stable edge property, a satisfying set will remain a satisfying set even when removing an edge that is not in the set from the graph. From this it easily follows that graph~$H_{m}$ contains exactly~$k^*$ edges, which form an edge set satisfying~$\Pi$. The edges that remain in graph~$H_m$ are therefore given as the output of the algorithm. It is easy to see that the algorithm takes polynomial time, given constant-time access to the oracle for~$\Q$. It only queries instances of size and parameter polynomial in~$n$.
\end{proof}

\subsection{\texorpdfstring{$k$}{k}-Path in restricted graph families} \label{section:kpath:turing}

Using the self-reduction algorithm presented in the previous section, we now give Turing kernels for \kPath on restricted graph families. The overall idea is the same as for \kCycle: the Turing kernelization repeatedly shrinks the graph by finding a suitable separation of order (at most) two and restricting its smaller side to only the vertices of a maximum-size witness, for each of the six ways in which a longest path crosses a separation. The separations are chosen so that the smaller side has size polynomial in~$k$, allowing the six witnesses to be computed by queries to the oracle whose size is polynomial in the parameter.

The main workhorse of the procedure is the subroutine given by Algorithm~\ref{alg:reduce-p}, which reduces a separation~$(A,B)$ after communication with an oracle for instances of size polynomial in~$|A|$. The subroutine is similar as that for the \kCycle case (Algorithm~\ref{alg:reduce-c}), but works for separations of order one or two.

\begin{algorithm}[t]
\caption{\textsc{Reduce-P}$(G', A, B, k)$} \label{alg:reduce-p}
\begin{algorithmic}[1]
\REQUIRE{$(A,B)$ is a separation in~$G'$ of order one or two.}
\ENSURE{The existence of a $k$-path in~$G'$ is reported, or the graph~$G'$ is updated by removing all but~$\Oh(k)$ vertices of~$A \setminus B$. Upon completion, the graph~$G'[A] - (A \cap B)$ has at most one connected component if~$|A \cap B| = 1$, and at most 12 connected components otherwise. If~$G'$ initially contained a $k$-path, then the deletions preserve this property.}
\vspace{0.2cm}
		\STATE Let~$Z := A \cap B$
		\IF{the \kPath oracle applied to~$(G'[A], k)$ answers \yes}\label{line:kpath:directquery}
			\STATE Report the existence of a $k$-path in~$G'$ and halt
		\ELSIF{$Z = \{x,y\}$ has cardinality two}
			\STATE Apply Lemma~\ref{lemma:selfreduce:stableproperty} to~$(G'[A], x, y)$, find~$\P_1, \ldots, \P_6 \subseteq G'[A]$ as in Lem.~\ref{lemma:paths:through:separation}\label{line:kpath:twosepwitness}
			\STATE Remove the vertices~$A \setminus (\bigcup_{i=1}^6 V(\P_i))$ from~$G'$
		\ELSIF{$Z = \{x\}$ has cardinality one}
			\STATE Apply Lemma~\ref{lemma:selfreduce:stableproperty} to~$(G'[A], x, x)$ to find a longest $x$-path~$\P_A$ in~$G'[A]$\label{line:kpath:onesepwitness}
			\STATE Remove the vertices~$A \setminus V(\P_A)$ from~$G'$
		\ENDIF
\end{algorithmic}
\end{algorithm}

\begin{lemma} \label{lemma:reduce-p:correct}
Algorithm~\ref{alg:reduce-p} satisfies its specifications. It works in polynomial time, when given constant-time access to an oracle for an arbitrary parameterized problem~$\Q$ whose underlying classical problem~$\widetilde{Q}$ is NP-complete. The oracle is queried for instances of size polynomial in~$|A|$.
\end{lemma}
\begin{proof}
Let us first discuss correctness of the procedure. If the oracle finds a $k$-path in~$G'[A]$, then reporting this fact is clearly correct. If~$G'[A]$ has no $k$-path, this gives us a bound on the maximum-size witness structures found in Lines~\ref{line:kpath:twosepwitness} and~\ref{line:kpath:onesepwitness}: each subgraph described in Lemma~\ref{lemma:paths:through:separation} consists of at most two paths, so if no $k$-path exists each witness structure has~$\Oh(k)$ vertices. If~$G'[A]$ has no $k$-path, removing all vertices of~$A$ except those in the witness structures preserves a $k$-path in~$G'$, if one exists. For separations of order one this follows from Observation~\ref{observation:paths:through:cutvertex}, while Lemma~\ref{lemma:paths:through:separation} justifies the order-two case. Since all vertices of~$A$ are removed except for those in constantly many witness structures, which have~$\Oh(k)$ vertices each, the size reduction claimed in the postcondition is achieved.

Let us consider the number of connected components of~$G'[A] - (A \cap B) = G'[A] - Z$ upon termination. If~$|Z| = 1$, then upon termination~$G'[A] - Z$ consists of the vertices~$V(\P_A)$ of the maximum-length $x$-path found by Lemma~\ref{lemma:selfreduce:stableproperty}, which clearly form at most one component. If~$|Z| = 2$, then upon termination~$G'[A] - Z$ is the graph induced by the vertices of the six different types of witness structures. Since each type of witness yields at most two connected components after removing~$Z = \{x,y\}$, it follows that~$G'[A] - Z$ has at most 12 connected components.

Since an execution consists of some simple graph manipulations, one oracle query, and one invocation of Lemma~\ref{lemma:selfreduce:stableproperty}, the running time of Algorithm~\ref{alg:reduce-p} is polynomial when given suitable oracle access. As the graph parameter to Lemma~\ref{lemma:selfreduce:stableproperty} is~$G'[A]$, the queries produced by the lemma are of size polynomial in~$|A|$. The type of oracle described in the statement of Lemma~\ref{lemma:reduce-p:correct} is compatible with what is required by Lemma~\ref{lemma:selfreduce:stableproperty}. In general, the direct oracle query in Line~\ref{line:kpath:directquery} can be transformed into a query to~$\Q$ using an NP-completeness transformation as in the proof of Lemma~\ref{lemma:selfreduce:stableproperty}. However, this is not needed in our Turing kernel applications: the oracle will be able to answer the \kPath query about the induced subgraph~$G'[A]$ directly.
\end{proof}

Algorithm~\ref{alg:reduce-p} is used in a bottom-up reduction procedure on a Tutte decomposition to obtain polynomial Turing kernels for \kPath on several restricted graph classes.

\begin{theorem} \label{theorem:kpath}
The \kPath problem has a polynomial-size Turing kernel when restricted to planar graphs, graphs of maximum degree~$t$, claw-free graphs, or $K_{3,t}$-minor-free graphs, for each constant~$t \geq 3$.
\end{theorem}
\begin{proof}
We will prove that, for each choice of restricted graph class~$\G$, an instance~$(G \in \G,k)$ of \kPath can be solved in polynomial time when given access to a constant-time oracle that decides \kPath for instances~$(H \in \G, k')$ in which~$|V(H)|$ and~$k'$ are bounded polynomially in~$k$. Since the \kPath problem is NP-complete for all graph classes in the theorem statement (cf.~\cite{LiCM00}), the classical language (in which the parameter is encoded in unary) underlying the parameterized \kPath problem restricted to~\G is NP-complete in all cases. We may therefore safely invoke Algorithm~\ref{alg:reduce-p} during the reduction algorithm.

Since a path is contained entirely within a single connected component, by running the algorithm separately on each connected component of the input graph we may assume that the input instance~$(G,k)$ is connected. 

\paragraph{Decompose} We compute a Tutte decomposition~$(T,\X)$ of~$G$~\cite{HopcroftT73}. Observe that if the circumference of a graph is~$k+1$, then it contains a $k$-path. Using Lemma~\ref{lemma:torsos:restricted}, the same argumentation as in Claim~\ref{claim:planar:decomposition:width} (but using a different polynomial bound, given by Theorem~\ref{theorem:circumference}) therefore justifies answering \yes if the width of~$(T,\X)$ exceeds some fixed polynomial in~$k$. If not, we make a copy~$G'$ of~$G$, a copy~$(T',\X')$ of the decomposition, and root~$T'$ at an arbitrary vertex.  

\paragraph{Query and reduce} The procedure that reduces the \kPath instance based on information obtained by oracle queries is given as Algorithm~\ref{alg:queryreducepath} on page~\pageref{alg:queryreducepath}. We use it in the same way as for \kCycle: we apply the reduction algorithm to the root node of the decomposition~$(T',\X')$ of~$G'$ with integer~$k$. If the procedure reports the existence of a $k$-path then the Turing kernelization answers \yes. If the procedure finishes without reporting a $k$-path, we query the \kPath oracle for the final reduced graph~$G'$ with parameter value~$k$. By the postcondition of the reduction algorithm, after it completes on the root node~$r$ the number of vertices that remain in the graph is~$\Oh(k \cdot (|E(\torso(G, \X(r)))| + |\X(r)|))$. Since~$|\X(r)|$ is bounded by a fixed polynomial in~$k$ (that depends on the graph class) in the decomposition phase and the size of a graph is obviously at most quadratic in its order, the queried instance~$(G',k)$ has size polynomial in~$k$. The answer of the oracle is given as the output of the Turing kernelization algorithm.

\begin{algorithm}[t]
\caption{\textsc{Kernelize-Path}$(G, G', (T', \X'), i, k)$} \label{alg:queryreducepath}
\begin{algorithmic}[1]
\REQUIRE{$G'$ is an induced subgraph of~$G$ with a tree decomposition~$(T',\X')$ of adhesion at most two. A node~$i$ of~$T'$ is specified.}
\ENSURE{The existence of a $k$-path in~$G$ is reported, or the graph~$G'$ and decomposition~$(T',\X')$ are updated by removing vertices of~$\X'(T'[i]) \setminus \X'(i)$, resulting in~$|\X'(T'[i])| \in \Oh(k \cdot (|E(\torso(G, \X(i)))| + |\X(i)|))$. If~$G'$ initially contained a $k$-path, then the deletions preserve this property.}
\vspace{0.2cm}
	\FOREACH{child~$j$ of~$i$ in~$T'$}
		\STATE Recursively execute \textsc{Kernelize-Path}$(G', (T', \X'), j, k)$
		\STATE Let~$Z := \X'(i) \cap \X'(j)$, let~$A := \X'(T'[j])$, and let~$B := (V(G') \setminus A) \cup Z$
		\STATE \textsc{Reduce-P}$(G', A, B, k)$\label{line:kpath:shrinkchild}
	\ENDFOR
	\FOREACH{vertex $x \in \X'(i)$}
		\WHILE{there are distinct children~$j_1, j_2$ of~$i$ in~$T'$ such that $\X'(i) \cap \X'(j_1) = \X'(i) \cap \X'(j_2) = \{x\}$}
			\STATE Let~$A := \X'(T'[j_1]) \cup \X'(T'[j_2])$, let~$B := (V(G') \setminus A) \cup \{x\}$
			\STATE \textsc{Reduce-P}$(G', A, B, k)$
		\ENDWHILE
	\ENDFOR
	\FOREACH{pair $\{x,y\} \in \binom{\X'(i)}{2}$}
		\WHILE{there are 13 distinct children~$j_1, \ldots, j_{13}$ of~$i$ in~$T'$ such that $\X'(i) \cap \X'(j_{s}) = \{x,y\}$ for all~$s \in [13]$}\label{line:kpath:removechild}
			\STATE Let~$A := \bigcup_{s=1}^{13} \X'(T'[j_s])$ and let~$B := (V(G') \setminus A) \cup \{x,y\}$
			\STATE \textsc{Reduce-P}$(G', A, B, k)$
		\ENDWHILE
	\ENDFOR
\end{algorithmic}
\end{algorithm}

Similar to the \kCycle Turing kernel, the procedure to reduce the subtree rooted at a node~$i$ has two stages. First it recursively shrinks subtrees rooted at the children~$j$ of~$i$. Afterward it reduces the number of children of~$i$, by considering sets of children that have the same adhesion to their parent bag~$i$, defining a separation based on them, and invoking Algorithm~\ref{alg:reduce-p}. Since that algorithm shrinks the number of connected components of~$G'[A] - \{x,y\}$ to at most 12, when there are 13 children with the same adhesion one of the child subtrees is guaranteed to disappear after such a reduction step. This shows that the \textbf{while}-loop of Line~\ref{line:kpath:removechild} terminates in polynomial time. As the recursive process consists of one bottom-up sweep over the Tutte decomposition, together with the bound for Algorithm~\ref{alg:reduce-p} given by Lemma~\ref{lemma:reduce-p:correct} this establishes the overall polynomial-time running time. The correctness of this approach follows by induction, using that Lemma~\ref{lemma:reduce-p:correct} guarantees that invocations of Algorithm~\ref{alg:reduce-p} preserve the existence of a $k$-path. The pairs~$(A,B)$ defined in the algorithm are valid separations by Propositions~\ref{proposition:separation:from:edge} and~\ref{proposition:separation:from:children}, since~$(T',\X')$ is invariantly a tree decomposition. To prove that the algorithm satisfies its specifications, it remains to prove the size bound claimed in the postcondition.

\begin{numberedclaim} \label{claim:path:postcondition:sizebound}
When the execution for node~$i$ terminates we have: $$|\X'(T'[i])| \in \Oh(k \cdot (|E(\torso(G, \X(i)))| + |\X(i)|)).$$
\end{numberedclaim}
\begin{claimproof}
By the postcondition and induction, the call to Algorithm~\ref{alg:reduce-p} in the first \textbf{for each} loop removes, for each child~$j$ of~$i$, all but~$\Oh(k)$ vertices of~$A \setminus B = \X'(T'[j]) \setminus \X'(i)$. Upon completion, each child subtree therefore represents~$\Oh(k)$ vertices of~$G'$ that are not in~$\X'(i)$ themselves. The second \textbf{for each} loop ensures that, upon termination, for each vertex~$x$ in~$\X'(i)$ there is at most one child of~$i$ whose adhesion to~$i$ is exactly~$\{x\}$. Similarly, the third \textbf{for each} loop ensures that the number of children with identical size-2 adhesions~$\{x,y\}$ is at most~$12$. We claim that all adhesions between~$i$ and a child~$j$ have size one or two, and hence that all children of~$i$ are accounted for in this way. To see that, observe that the execution for node~$i$ does not remove vertices from the bag of node~$i$ or its ancestors. Hence during the execution for node~$i$ the adhesion of~$i$ to its children in~$(T',\X')$ equals the adhesion in the original Tutte decomposition~$(T,\X)$. Since we started from a connected graph~$G$, each adhesion has size at least one. By the properties of a Tutte decomposition, each adhesion has size at most two. Hence each child of~$i$ has an adhesion of one or two to~$i$. Since each nonempty adhesion in a Tutte decomposition is a minimal separator, it follows that for each child~$j$ of~$i$ with a size-2 adhesion~$\{x,y\}$, the set~$\{x,y\}$ is a minimal separator in the original graph~$G$. By Proposition~\ref{proposition:minimalsep:in:tutte:edge:in:torso}, for each child with a size-$2$ adhesion~$\{x,y\}$, the corresponding pair is connected by an edge. It follows that, upon termination for node~$i$, the number of children with a size-$2$ adhesion is at most~$12 |E(\torso(G, \X(i)))|$. The number of children with a size-$1$ adhesion is at most~$|\X(i)|$. The application of \textsc{Reduce-P} in Line~\ref{line:kpath:shrinkchild} ensures that for each child~$j$, all but~$\Oh(k)$ vertices of~$\X'(T'[j]) \setminus \X'(i)$ are removed. Hence each child contributes~$\Oh(k)$ vertices to~$\X'(T'[i]) \setminus \X'(i)$. As we just argued that the number of children of~$i$ upon termination is bounded by~$\Oh(|E(\torso(G, \X(i)))| + |\X'(i)|)$, while the bag of~$i$ contributes another~$\X'(i) = \X(i)$ vertices, the claim follows.
\end{claimproof}

Claim~\ref{claim:path:postcondition:sizebound} shows that Algorithm~\ref{alg:queryreducepath} satisfies its postcondition, if the input satisfies the precondition. With the previous argumentation, this shows that the algorithm runs in polynomial time with access to an oracle for answering queries on instances of size polynomial in~$k$. This concludes the proof of Theorem~\ref{theorem:kpath}.
\end{proof}

\section{Constructing solutions} \label{section:constructing:solutions}

Motivated by the definition of Turing kernelization, we have presented our results in terms of decision problems where the goal is to answer \yes or \no; this also simplified the presentation. In practice, one might want to \emph{construct} long paths or cycles rather than merely report their existence. Our techniques can be adapted to construct a path or cycle of length at least~$k$, if one exists.

\begin{corollary}
For each graph class~$\G$ as described in Theorem~\ref{theorem:kpath}, there is an algorithm that, given a pair~$(G \in \G,k)$ either outputs a $k$-cycle (respectively $k$-path) in~$G$, or reports that no such object exists. The algorithm runs in polynomial time when given constant-time access to an oracle that decides \kCycle (respectively \kPath) on~$\G$ for instances of size and parameter bounded by some polynomial in~$k$ (whose degree depends on~$\G$).
\end{corollary}
\begin{proof} We treat the cases of paths and cycles consecutively, starting with paths.

\paragraph{Constructing paths} Let us first consider the \kPath case. If the Turing kernelization outputs \yes because a \kPath oracle gives a \yes answer on an instance~$(G'[A], k)$ of size polynomial in~$k$, then a straight-forward self-reduction on this small instance using Lemma~\ref{lemma:selfreduce:stableproperty} can be used to construct a solution (the lemma guarantees that the oracle is only queried for instances of size polynomial in~$|V(G'[A])|$). However, the situation is more complicated when the Turing kernelization answers \yes based on Theorem~\ref{theorem:circumference} because there is a large bag in the Tutte decomposition: applying Lemma~\ref{lemma:selfreduce:stableproperty} to the torso of the bag would violate the size bound on the queried instances, since the torso can be arbitrarily large. For triconnected claw-free graphs~\cite[\S5.3]{BilinskiJMY11} and bounded-degree graphs~\cite[\S6]{ChenGYZ06}, polynomial-time algorithms are known that construct a path of length~$n^{\alpha}$ for some positive~$\alpha$, which can be used to construct a $k$-path if the width of the Tutte decomposition exceeds~$k^{1/\alpha}$. No such algorithmic results are known for planar or $K_{3,t}$-minor-free graphs. However, for these graph classes we can construct long paths by exploiting the fact that they are closed under edge deletions, through a self-reduction that calls the Turing kernelization algorithm, as follows. 

Let~$(G,k)$ be a planar or $K_{3,t}$-minor-free instance that contains a $k$-path. Order the edges of~$G$ as~$e_1, \ldots, e_m$ and set~$G_0 := G$. For~$i \in [m]$ we apply the \emph{Turing kernelization} to the instance~$(G_{i-1} - e_i, k)$. If it outputs \yes then we set~$G_i := G_{i-1} - e_i$, otherwise we set~$G_i := G_{i-1}$. After~$m$ calls to the Turing kernelization the resulting graph~$G_m$ contains exactly the edges of a $k$-path. The fact that planar and $K_{3,t}$-minor-free graphs are closed under edge deletions ensures that we may safely apply the Turing kernelization to all graphs~$G_i$. Since the Turing kernelization only queries the oracle for instances of size polynomial in~$k$, we obtain an algorithm that constructs a $k$-path in polynomial time when given access to a \kPath oracle for the restricted graph class. The oracle is only invoked for instances of size and parameter polynomial in~$k$.

\paragraph{Constructing cycles} We move on to the \kCycle case. If the Turing kernelization outputs \yes because the oracle gives this answer on a small instance~$G'[A]$, or because the Tutte decomposition has a bag of large width, then we may proceed similarly as in the \kPath case to construct a solution. However, there is an extra complication for \kCycle since the Turing kernelization may output \yes because its call to Lemma~\ref{lemma:selfreduce:xypath} reports the existence of a long $xy$-path for some minimal separator~$\{x,y\}$. Note that Lemma~\ref{lemma:selfreduce:xypath} is only applied to instances~$(G'[A], k, x, y)$ of size polynomial in~$k$. If the existence of a long $xy$-path is reported, we can therefore use the self-reduction of Lemma~\ref{lemma:selfreduce:stableproperty} to \emph{construct} the edge set of a maximum-length $xy$-path~$\P$ by using oracle queries of size polynomial in~$|V(G'[A])|$ (which is polynomial in~$k$). The proof of Proposition~\ref{proposition:xyseparator:path:gives:cycle} easily yields a polynomial-time algorithm to complete~$\P$ into a $k$-cycle, which handles this last case and completes the proof.
\end{proof}

\section{Multicolored paths in bounded-degree graphs} \label{section:multicolored}

An input for the \kMulticoloredPath problem consists of a graph~$G$, an integer~$k$ and a (generally not proper) coloring~$f \colon V(G) \to [k+1]$ of its vertices. The question is whether there is a path of length~$k$ (which spans~$k+1$ vertices) that contains exactly one vertex of each color. Hermelin et al.~\cite{HermelinKSWW15} showed that \kMulticoloredPath is WK[1]-complete under polynomial-parameter transformations. They conjectured that WK[1]-hard problems do not have polynomial-size Turing kernels. We show that the multicolored problem remains WK[1]-complete even for subcubic graphs.

\begin{theorem} \label{theorem:multicolored}
The \kMulticoloredPath problem on graphs of maximum degree at most three is WK[1]-complete.
\end{theorem}
\begin{proof}
Membership in WK[1] is implied by the fact that the general version of the problem (without the degree bound) is contained in WK[1], as shown by Hermelin et al.~\cite[Lemma 20]{HermelinKSWW15}. To prove hardness for WK[1], we give a polynomial-parameter transformation~\cite[Definition 3]{HermelinKSWW15} from the WK[1]-hard~\cite[Theorem 5]{HermelinKSWW15} \nExactSetCover problem to \kMulticoloredPath on bounded degree graphs.

Consider an input~$(\F,U)$ of \nExactSetCover with consists of a size-$m$ set family~$\F$ over a finite universe~$U$ of size~$n$. The question is whether there is a subfamily~$\F' \subseteq \F$ such that each element of~$U$ is contained in exactly one set of~$\F'$. By duplicating some sets in~$\F$, which does not increase the instance size by more than two, we may assume that~$|\F| = 2^r - 1$ for an integer~$r$.

If~$m \geq 2^n$ then the straight-forward (cf.~\cite[Theorem 6.1]{CyganFKLMPPS15}) dynamic program for \nExactSetCover over subsets of the universe, which runs in~$\Oh(2^n (n+m)^{\Oh(1)}) \subseteq \Oh(m \cdot (n+m)^{\Oh(1)})$ time, takes time polynomial in~$n+m$ and therefore in the input size. We may apply it and output a constant-size instance with the same answer. In the remainder we may therefore assume that~$\log m \leq n$ which implies that we may afford to increase the parameter by a polynomial in~$\log m$.

The instance of \kMulticoloredPath that we construct consists of~$2(n-1)$ complete binary trees~$O_1, I_2, O_2, I_3, \ldots, O_{n-1}, I_n$ with~$2^r$ leaves each. We color the vertices such that all vertices that belong to a common level of a common tree have the same color, while vertices of different trees or on different levels have different colors. Since a complete binary tree with~$2^r$ leaves consists of~$r + 1$ levels, this requires~$2(n-1)(r+1)$ different colors. We also create a unique color~$c(u_1), \ldots, c(u_n)$ for each element of~$U$. We connect the root of tree~$I_i$ to the root of tree~$O_i$ for all~$2 \leq i \leq n-1$. We encode the sets of the instance as follows. For each~$j \in [2^r - 1]$, for each~$i \in [n-1]$, we do the following. Let~$F_j = \{u_{i_1}, \ldots, u_{i_\ell}\}$ be the $j$th set in~$\F$. Create a path on~$\ell$ vertices and give the $a$th vertex on this path color~$c(u_{i_a})$. We make the first vertex of the path adjacent to the $j$th leaf of~$O_i$ and the $j$th leaf of~$I_{i+1}$. Additionally, we make the~$2^r$th leaf of~$O_i$ adjacent to the $2^r$th leaf of~$I_i$. After doing this for each choice of~$i$ and~$j$ we output the resulting colored graph with the parameter~$k' := n + (2(n-1)(r+1)) - 1 \in \Oh(n \cdot r) \in \Oh(n \log m) \in \Oh(n^2)$. It is easy to see that the construction can be performed in polynomial time and that the parameter~$k'$ is suitably bounded for a polynomial-parameter transformation. The maximum degree of the resulting instance is three since it is obtained by gluing paths to the leaves of binary trees.

It remains to prove that~$(\F,U)$ has an exact set cover if and only if~$G$ has a multicolored $k$-path. In one direction, suppose that there is an exact set cover with~$\ell$ sets~$F_1, \ldots, F_\ell$; observe that~$\ell \leq n$ since each set contains at least one element and no element is allowed to be covered twice. Construct a multicolored path starting from the root of~$O_1$, moving down the tree to the leaf corresponding to set~$F_1$, traverse the path to the corresponding leaf of~$I_2$, move up the tree to the root of~$I_2$, traverse the edge to the root of~$O_2$, move down the tree to the leaf that corresponds to~$F_2$, and so on. After the~$\ell$ sets have all been used, traverse through the remaining trees to the root of~$I_n$ by using the direct connection between the~$2^r$th leaves of the relevant trees. If the sets cover~$U$ exactly, then, since the elements on the used subpaths correspond to the colors of the universe elements, while one vertex of each level of each binary tree is used, the resulting $k'$-path is multicolored.

The other direction can be proven similarly. Since each color can be used only once by a path, the linear structure of the instance forces a multicolored $k'$-path to start at the root of~$O_1$, traverse down to a leaf, and use a connection that either corresponds to a set of~$\F$ or to skipping a set. To use all the available colors once, the path has to traverse subpaths corresponding to sets that cover the universe exactly. This concludes the proof of Theorem~\ref{theorem:kpath}.
\end{proof}

The theorem shows that the \kMulticoloredPath problem remains WK[1]-hard on bounded-degree graphs. However, Theorem~\ref{theorem:kpath} shows that the uncolored \kPath problem admits a polynomial Turing kernel on bounded-degree graphs. This indicates that the colored problem may be significantly harder to preprocess than the uncolored version.

\section{Conclusion} \label{section:conclusion}
We presented polynomial-size Turing kernels for \kPath and \kCycle on restricted graph families using the \emph{Decompose-Query-Reduce} framework, thereby answering an open problem posed by Lokshtanov~\cite{Lokshtanov09} and Misra et al.~\cite{MisraRS11}. Our results form the second~\cite{ThomasseTV14} example of adaptive Turing kernelization of polynomial size.

The question remains whether \kPath admits a polynomial-size Turing kernel in general graphs. Theorem~\ref{theorem:multicolored} indicates that the WK[1]-hardness of \kMulticoloredPath~\cite[Theorem 7]{HermelinKSWW15} may not be relevant for the \kPath problem, suggesting the possibility of a positive answer. Significant new ideas will be needed to solve this case in the positive. The Tutte decomposition employed here is of little use in general graphs, since the elementary building blocks of the decomposition (triconnected graphs) do not yield anything useful. While triconnected \emph{planar} graphs have circumference~$\Oh(n^{\alpha})$ for a positive constant~$\alpha$, the circumference of a general triconnected graph may be as low as~$\Oh(\log n)$, which is achieved by considering the join of a triangle with a complete binary tree. Different decomposition methods may be used in general graphs; for example, in linear time one can either find a $k$-path or establish that the treedepth (and therefore treewidth) is at most~$k$, which gives a decomposition of the graph as an embedding into the closure of a rooted tree of height~$k$ (cf.~\cite[Theorem 8.2]{DowneyF99}). However, since the adhesion of the corresponding tree decomposition can be linear in~$k$, this does not seem as useful for identifying irrelevant parts of the input. Analyzing \kPath on chordal graphs may be an intermediate step: the example above shows that even for triconnected chordal graphs the circumference may be~$\Oh(\log n)$.

Our results also prompt the investigation of other subgraph and minor testing problems. For example, does the problem of testing whether a planar graph~$G$ has a subgraph isomorphic to~$H$ admit a polynomial Turing kernel, parameterized by~$|H|$? The simplest unresolved case of this problem seems to be the \textsc{Exact $k$-Cycle} problem of finding a cycle of length \emph{exactly}, rather than at least,~$k$. The present approach fails on this problem since it is already unclear how to deal with triconnected planar graphs. Similar questions can be asked for the problem of finding a graph~$H$ as a minor in a planar graph~$G$, parameterized by~$|H|$. To further understand the nature of Turing kernelization, one might also investigate whether the adaptive Turing kernel given here can be transformed into a non-adaptive Turing kernel, whose queries only depend on the input and not on the answers to earlier queries. Since the queries in a non-adaptive Turing kernel can be executed in parallel, this might offer practical advantages.

\textbf{Acknowledgments}. We are grateful to Micha\l \ Pilipczuk for suggesting Theorem~\ref{theorem:multicolored} and D\'{a}niel Marx for suggesting its current easy proof.

\bibliography{../Paper}
\bibliographystyle{abbrvurl}

\newpage
\appendix

\section{Proof of Theorem \ref{theorem:tutte}} \label{appendix:tutte:decomposition}

\begin{untheorem}
For every graph~$G$ there is a tree decomposition~$(T,\X)$ of adhesion at most two, called a \emph{Tutte decomposition}, such that:
\begin{enumerate}
	\item for each node~$i \in V(T)$, the graph~$\torso(G, \X(i))$ is a triconnected topological minor of~$G$, and
	\item for each edge~$\{i,j\}$ of~$T$ the set~$\X(i) \cap \X(j)$ is a minimal separator in~$G$ or the empty set.
\end{enumerate}
\end{untheorem}

\begin{proof}
The proof uses induction on the order of~$G$ and a case distinction on the connectivity of~$G$.

\paragraph{Triconnected} The base case of the induction is when~$G$ is a triconnected graph. Note that, by our definition, the single-vertex graph is triconnected. The trivial tree decomposition~$(T,\X)$ where~$T$ consists of a single node~$i$ and~$\X(i) = V(G)$ is a Tutte decomposition in this case. The adhesion is zero since~$T$ has no edges, trivially satisfying~\ref{tutte:minseparators}. The graph~$\torso(G, \X(i))$ coincides with~$G$, which is triconnected by assumption and a topological minor of~$G$ by definition.

For the induction step, we assume that the statement is true for all graphs of order less than~$|V(G)|$ and proceed by a case distinction on the connectivity.

\paragraph{Disconnected} If~$G$ is disconnected, then let~$C_1, \ldots, C_t$ be its connected components. Since each component has fewer vertices than~$G$ itself, by induction there are Tutte decompositions~$(T_1, \X_1), \ldots, (T_t, \X_t)$ of each connected component. Obtain a Tutte decomposition~$(T,\X)$ of~$G$ as follows. The tree~$T$ is obtained from the forest~$T_1 \cup \ldots \cup T_t$ by adding arbitrary edges to make the forest connected. Each node~$i$ of~$T$ belongs to a unique tree~$T_j$; the associated bag~$\X(i)$ is simply~$\X_j(i)$.

We claim that~$\torso(C_j, \X_j(i)) = \torso(G, \X(i))$ for all~$i,j$. This follows from the fact that~$C_j[\X_j(i)] = G[\X(i)]$, that all paths in~$C_j$ connecting vertices~$u,v \in \X_j(i)$ with interior vertices that avoid~$\X_j(i)$ also exist in~$G$ (since~$C_j$ is a subgraph of~$G$), and that no such paths exist in~$G$ that do not exist in~$C_j$, since no vertex of~$G - C_j$ can be reached from a vertex in~$C_j$ since it belongs to a different connected component. Hence the torso of each bag of~$(T,\X)$ is a triconnected topological minor of a connected component of~$G$ (by induction), and therefore of~$G$ itself. It is easy to verify that~$T$ is a tree decomposition of adhesion at most two. Each nonempty intersection of the bags of adjacent nodes in the resulting decomposition was also an intersection in one of the Tutte decompositions for the connected components; hence the intersection forms a minimal separator in one of the connected components by induction.

\paragraph{Cut vertex} Assume that~$G$ is connected but contains a cut vertex~$v$. Let~$C_1, \ldots, C_t$ be the connected components of~$G - v$, and for each~$i \in [t]$ let~$C'_i := G[V(C_i) \cup \{v\}]$. Each edge of~$G$ is contained in exactly one graph~$C'_i$. By induction there are Tutte decompositions~$(T_1, \X_1), \ldots, (T_t, \X_t)$ of the graphs~$C'_1, \ldots, C'_t$. Since each graph~$C'_i$ contains vertex~$v$, each tree decomposition~$T_i$ has a node~$n_i$ such that~$v \in \X_i(n_i)$. The tree~$T$ of the Tutte decomposition~$(T,\X)$ is obtained from the forest~$T_1 \cup \ldots \cup T_t$ by adding an edge between~$n_i$ and~$n_1$ for each~$i \geq 2$; the bags of~$T$ correspond to the bags of the individual decompositions~$T_i$ as before. 

To see that~$\torso(G, \X(i))$ is a triconnected topological minor of~$G$ for each~$i \in V(T)$, we prove that~$\torso(C'_j, \X_j(i)) = \torso(G, \X(i))$ for each~$j \in [t]$ and~$i \in V(T_j)$. By transitivity of topological minors this suffices to prove the claim, since each~$\torso(C'_j, \X_j(i))$ is a triconnected topological minor of~$C'_j$ by induction, while~$C'_j$ is a subgraph of~$G$. To see that~$\torso(C'_j, \X_j(i)) = \torso(G, \X(i))$, observe the following three facts. 
\begin{enumerate}
	\item $C'_j[\X_j(i)] = G[\X(i)]$ by our choice of~$\X$.
	\item Each path in~$C'_j$ whose interior vertices avoid~$\X_j(i)$ and that connects two vertices in~$\X_j(i)$, also exists in~$G$ (since~$C'_j$ is a subgraph of~$G$).
	\item All paths in~$G$ connecting two vertices in~$\X_j(i)$ whose interior vertices avoid~$\X_j(i)$ lie entirely within~$C'_j$. This follows from the fact that~$v$ is a cut vertex, which implies that the vertices of~$G - C'_j$ are not adjacent to any vertex of~$C'_j$ except~$v$. 
\end{enumerate}
Hence each torso of~$(T, \X)$ equals a torso of a Tutte decomposition of one of the components~$C'_i$, and is therefore a triconnected topological minor of~$G$ by induction. Again it is easy to verify that the resulting structure is a tree decomposition of adhesion at most two. The only new edges introduced into the decomposition tree are those to connect the various trees together; the intersection of such bags is the cut vertex~$v$ which is a minimal separator.

\paragraph{Separation pair} Finally, assume that~$G$ is connected and contains no cut vertices, but contains a separation pair~$\{u,v\}$. Let~$C_1, \ldots, C_t$ be the connected components of~$G - \{u,v\}$. As~$G$ contains no cut vertices we know that~$\{u,v\}$ is a minimal separator. For each~$i \in [t]$ let~$C'_i$ be the graph obtained from~$G[V(C_i) \cup \{u,v\}]$ by adding the edge~$\{u,v\}$ if it did not exist already.

\begin{numberedclaim} \label{claim:tutte:path:through:ci}
For each~$i \in [t]$ there is a $uv$-path in~$G$ whose internal vertices all belong to~$C'_i$.
\end{numberedclaim}
\begin{claimproof}
Since~$G$ is connected, component~$C_i$ is adjacent to at least one of~$u$ and~$v$. If~$C_i$ is not adjacent to both of them, then one of~$\{u,v\}$ is a cut vertex. As we are in the case that~$G$ has no cut vertex, we therefore know that~$C_i$ contains both a neighbor~$u'$ of~$u$ and a neighbor~$v'$ of~$v$. There is a $u'v'$-path in~$C_i$, since~$C_i$ is connected. Together with~$u$ and~$v$ this gives the desired path.
\end{claimproof}

Since~$\{u,v\}$ is a separation pair, there are at least two components ($t \geq 2$), prompting the following observation.

\begin{observation} \label{observation:tutte:path:avoiding:ci}
For each~$i \in [t]$ there is a $uv$-path in~$G$ whose internal vertices avoid~$C'_i$.
\end{observation}

\begin{numberedclaim} \label{claim:tutte:path:shortcut}
Let~$p, q \in V(C'_i)$ for some~$i \in [t]$ with~$p \neq q$, and let~$\P$ be a $pq$-path in~$G$. Then there is a $pq$-path~$\P'$ in~$C'_i$ such that~$V(\P') \subseteq V(\P)$.
\end{numberedclaim}
\begin{claimproof}
If~$V(\P) \subseteq V(C'_i)$ then the claim is trivial, so assume that~$\P$ contains a vertex~$r \not \in V(C'_i)$. By the separation property of~$\{u,v\}$, all paths from~$p$ or~$q$ to~$r$ pass through~$u$ or~$v$ (even if~$\{p,q\} \cap \{u,v\} \neq \emptyset$). It follows that when traversing~$\P$ from~$p$ to~$r$ we pass through one of~$\{u,v\}$, and when traversing~$\P$ from~$r$ to~$q$ we pass through the other. Hence~$\P$ contains a $uv$-subpath, and all vertices not on this~$uv$-subpath must belong to~$C'_i$ since the separator~$\{u,v\}$ has size two. As there is a direct edge between~$u$ and~$v$ in~$C'_i$, we can replace the $uv$-subpath by the direct edge to obtain a $pq$-path~$\P'$ in~$C'_i$ as desired.
\end{claimproof}

Using these claims we proceed with the proof. Let~$(T_1, \X_1), \ldots, (T_t, \X_t)$ be Tutte decompositions of~$C'_1, \ldots, C'_t$, which exist by induction. Since each graph~$C'_i$ contains the edge~$\{u,v\}$, by property~\ref{td:covere} of Definition~\ref{def:treedec} each decomposition~$(T_i, \X_i)$ has a node~$n_i$ whose bag~$\X_i(n_i)$ contains both~$u$ and~$v$. As in the previous case, the tree~$T$ of the Tutte decomposition~$(T,\X)$ of~$G$ is obtained from~$T_1 \cup \ldots \cup T_t$ by adding the edges~$\{n_1, n_i\}$ for all~$i \geq 2$. These edges do not increase the adhesion of the decomposition beyond two since the bags corresponding to their endpoints have an intersection of size exactly two consisting of~$u$ and~$v$; these are the only vertices occurring in more than one graph~$C'_i$. 

\begin{numberedclaim} \label{claim:tutte:cprime:topminor}
For each~$i \in [t]$ the graph~$C'_i$ is a topological minor of~$G$.
\end{numberedclaim}
\begin{claimproof}
By Observation~\ref{observation:tutte:path:avoiding:ci} there is a $uv$-path in~$G$ whose internal vertices avoid~$C'_i$. Since we can shortcut this path wherever possible without increasing the set of visited vertices, this implies that there is also an \emph{induced} $uv$-path in~$G$ whose internal vertices avoid~$C'_i$, say~$\P$. We can obtain~$C'_i$ from the graph~$G[V(C'_i) \cup V(\P)]$ as follows: for each interior vertex of~$\P$, remove all its incident edges except those to its predecessor and successor on~$\P$. Afterward, repeatedly replace the resulting degree-2 interior vertices of~$\P$ by direct edges, thereby creating a direct edge between~$u$ and~$v$. Hence~$C'_i$ can be built by the legal operations for taking topological minors.
\end{claimproof}

\begin{numberedclaim} \label{claim:tutte:gsubgraphprime}
For each~$j \in [t]$ and~$i \in V(T_j)$ the graph~$\torso(G, \X(i))$ is a subgraph of~$\torso(C'_j, \X_j(i))$.
\end{numberedclaim}
\begin{claimproof}
Recall that~$\X(i) = \X_j(i)$. Consider an edge~$\{p,q\}$ of $\torso(G, \X(i))$; we prove the edge is also contained in~$\torso(C'_j, \X(i))$. If~$\{p,q\} \in E(C'_j)$ then this is trivial. If~$\{p,q\} \not \in E(C'_j)$ then in particular we know that~$\{p,q\} \not \in E(G)$. By definition of torso there must be a $pq$-path~$\P$ in~$G$ whose internal vertices avoid~$\X(i)$. By Claim~\ref{claim:tutte:path:shortcut} this implies the existence of a $pq$-path~$\P'$ in~$C'_j$ on a subset of the vertices of~$\P$, implying that the internal vertices of~$\P'$ avoid~$\X(i)$. Hence~$\{p,q\}$ is an edge of~$\torso(C'_j, \X_j(i))$.
\end{claimproof}

\begin{numberedclaim} \label{claim:tutte:primesubgraphg}
For each~$j \in [t]$ and~$i \in V(T_j)$ the graph~$\torso(C'_j, \X_j(i))$ is a subgraph of~$\torso(G, \X(i))$.
\end{numberedclaim}
\begin{claimproof}
Consider an edge~$\{p,q\}$ of~$\torso(C'_j, \X(i))$; we prove it is also contained in $\torso(G, \X(i))$. If~$\{p,q\}$ is an edge of~$C'_j$ different from~$\{u,v\}$, then by construction of~$C'_j$ this edge is also contained in~$G$ and therefore in the torso. If~$\{p,q\} = \{u,v\}$, then by Observation~\ref{observation:tutte:path:avoiding:ci} there is a $uv$-path in~$G$ whose internal vertices avoid~$C'_j$ and therefore avoid~$\X(i)$. This ensures~$\{p,q\} = \{u,v\}$ is an edge of~$\torso(G, \X(i))$. Hence all edges of~$C'_j$ are present in~$\torso(G, \X(i))$. 

All other edges~$\{p,q\}$ of~$\torso(C'_j, \X(i))$ were added on account of a $pq$-path~$\P$ through~$C'_j$ whose internal vertices avoid~$\X(i)$. If such a path~$\P$ does not use the edge~$\{u,v\}$ then it is also a path in~$G$; otherwise we can replace the direct edge~$\{u,v\}$ on~$\P$ by a $uv$-path whose internal vertices avoid~$C'_j$ and therefore~$\X(i)$, by Observation~\ref{observation:tutte:path:avoiding:ci}. In both cases we conclude there is a $pq$-path~$\P$ in~$G$ whose internal vertices avoid~$\X(i)$, proving that~$\{p,q\}$ is an edge of~$\torso(G, \X(i))$.
\end{claimproof}

Using the three claims we can prove that each graph~$\torso(G, \X(i))$ is a triconnected topological minor of~$G$. Claims~\ref{claim:tutte:gsubgraphprime} and~\ref{claim:tutte:primesubgraphg} show that each torso of~$(T,\X)$ with respect to~$G$ is equal to a torso of~$(T_j, \X_j)$ with respect to~$C'_j$. By induction, the graphs~$\torso(C'_j, \X_j(i))$ are topological minors of~$C'_j$, which itself is a topological minor of~$G$ by Claim~\ref{claim:tutte:cprime:topminor}. By transitivity of topological minors we therefore establish Property~\ref{tutte:torsos}. It remains to establish Property~\ref{tutte:minseparators}.

\begin{numberedclaim} \label{claim:tutte:minimalseparator}
For each edge~$\{a,b\} \in E(T)$, the set~$S := \X(a) \cap \X(b)$ is a minimal separator in~$G$.
\end{numberedclaim}
\begin{claimproof}
For edges~$\{a,b\}$ of~$T$ that were added to connect the different decomposition trees together, note that~$\{u,v\} = \X(a) \cap \X(b)$ is the separation pair that defines this case. It is a minimal separator by the assumption that~$G$ does not have a cut vertex. In the remainder we consider an edge~$\{a,b\}$ of~$T$ that originates from one of the trees~$T_i$ that were obtained by induction, implying that~$S$ is a minimal separator in some~$C'_i$ for~$i \in [t]$. As the adhesion is at most two,~$|S| \leq 2$. Assume for a contradiction that~$S$ is not a separator in~$G$. Let~$p,q \in V(C'_i)$ be two vertices that lie in different connected components of~$C'_i - S$, but which are connected by a path~$\P$ in~$G - S$. By Claim~\ref{claim:tutte:path:shortcut}, the part of~$\P$ outside~$C'_i$ forms a~$uv$-path that can be replaced by the direct edge~$\{u,v\}$ in~$C'_i$ to obtain a $pq$-path entirely within~$C'_i$ while avoiding~$S$. But then~$C'_i - S$ contains a $pq$-path; a contradiction. Hence~$S$ is a separator in~$G$. To see that~$S$ is a \emph{minimal} separator, observe that~$G$ does not have any cut vertices by the case distinction. Hence no single vertex is a separator in~$G$. Since~$|S| \leq 2$ by the adhesion bound, no strict subset of~$S$ is a separator, implying minimality.
\end{claimproof}

Claim~\ref{claim:tutte:minimalseparator} establishes Property~\ref{tutte:minseparators} and concludes the case of the induction step that~$G$ has a separation pair. Since a connected graph without cut vertices or separation pairs is triconnected, any graph that is not covered by one of these cases is triconnected. It is therefore covered by the base case, which concludes the proof of Theorem~\ref{theorem:tutte}.
\end{proof}

\end{document}